\documentclass[usletter,titlepage,12pt]{amsart}

\usepackage{amsmath,amsthm,amssymb,authblk,amsfonts}
\usepackage[foot]{amsaddr}

\usepackage{natbib}

\usepackage{comment,enumerate,graphicx,hhline}

\usepackage[usenames,dvipsnames]{color, xcolor}

\usepackage{subcaption}
\usepackage{hyperref}

\usepackage{geometry}
\theoremstyle{plain}

\newtheorem{proposition}{Proposition}

\newtheorem{theorem}{Theorem}
\newtheorem{lemma}{Lemma}

\newtheorem*{claim*}{Claim}
\newtheorem{claim}{Claim}

\theoremstyle{remark}
\newtheorem{remark}{Remark}
\newtheorem{example}{Example}

\newtheoremstyle{named}{}{}{\itshape}{}{\bfseries}{.}{.5em}{#1\thmnote{
    #3}}
\theoremstyle{named}

\newtheoremstyle{named2}{}{}{\itshape}{}{\bfseries}{:}{.5em}{#1\thmnote{
    #3}}
\theoremstyle{named2}

\def\A{\mathcal{A}}

\def\H{\mathcal{H}}
\def\X{\mathcal{X}}
\def\C{\mathcal{C}}

\def\Re{\mathbf{R}}

\def\ep{\varepsilon}
\def\w{\omega} \def\W{\Omega}
\def\ta{\theta}
\def\al{\alpha}
\def\la{\lambda}
\def\da{\delta}
\def\g{\gamma} 
\def\phi{\varphi}

\def\one{\mathbf{1}}
\def\co{\mbox{co}}
\def\lcs{\mbox{lcs}}
\def\argmax{\mbox{argmax}} \def\argmin{\mbox{argmin}}
\def\supp{\mbox{supp}}

\def\ul{\underline}

\def\citeapos#1{\citeauthor{#1}'s (\citeyear{#1})}

\def\os{\emptyset}

\newcommand{\df}[1]{\textit{#1}}
\newcommand{\norm}[1]{\| #1 \|}

\newcommand{\abs}[1]{ \left | #1 \right | }

\def\nobjects{L}
\def\nagents{N}

\newcommand{\jz}[1]{\noindent \normalfont \textcolor{Red}{[JZ: #1]}}

\begin{document}

\author[Echenique]{Federico Echenique}
\author[Miralles]{Antonio Miralles}
\author[Zhang]{Jun Zhang}

\address[Echenique]{Division of the Humanities and Social Sciences,
  California Institute of Technology}
\address[Miralles]{Department of Economics, Universitat
  Aut\`{o}noma de Barcelona and Barcelona Graduate School of Economics and Universit\`{a} degli Studi di Messina.}
\address[Zhang]{Institute for Social and Economic Research, Nanjing Audit University.}

\title[Constraints]{Constrained Pseudo-market Equilibrium}

\thanks{We thank Eric Budish, Fuhito Kojima, Andy McLennan, Herv\'e Moulin, and Tayfun S\"onmez for comments.  Echenique thanks the National Science Foundation for its support through the grants SES-1558757 and CNS-1518941, and the Simons Institute at UC Berkeley for its hospitality while part of the paper was written. Miralles acknowledges financial support from the Spanish Ministry of Economy and Competitiveness, through the Severo Ochoa Programme for Centers of Excellence in R\&D (SEV-2015-0563). Zhang acknowledges financial support from National Natural Science Foundation of China (Grant No. 71903093).}

\date{\today}

\begin{abstract}
	We propose a pseudo-market solution to resource allocation problems subject to constraints. Our treatment of constraints is general: including bihierarchical constraints due to considerations of diversity in school choice, or scheduling in course allocation; and other forms of constraints needed to model, for example, the market for roommates, and combinatorial assignment problems. Constraints give rise to pecuniary externalities, which are internalized via prices. Agents pay to the extent that their purchases affect the value  of relevant constraints at equilibrium prices. The result is a constrained efficient market equilibrium outcome. The outcome is fair whenever the constraints do not single out individual agents. Our result can be extended to economies with endowments, and address participation constraints.
\end{abstract}



\maketitle

\tableofcontents
\newpage

\section{Introduction}
We analyze the use of pseudo-markets for assignment problems {\em under constraints} in market design environments where resources are indivisible and monetary transfers are forbidden. A pseudo-market is an artificial marketplace where agents are given fixed budgets of ``funny money'' that is only useful within the marketplace. Agents use artificial money to buy affordable probability shares of preferred goods at market clearing prices. By assigning probability shares to agents, it is possible to get (ex-ante) fairness among agents. The pseudo-market idea was first proposed by \cite{HZ1979} to solve the allocation of indivisible goods to an equal number of agents under unit demand and unit supply constraints. By pricing goods, Hylland and Zeckhauser (HZ) prove that market clearing prices exist, and the equilibrium outcome is efficient and fair.\footnote{If there are only supply constraints, efficiency is the result of the first welfare theorem. In HZ's model, however, there are also unit demand constraints and the first welfare theorem fails.}

We generalize and expand the scope of applicability of pseudo-markets. We consider assignment problems under various constraints. These problems include familiar constrained assignment problems, such as job assignments under regional ``ceiling'' and ``floor'' quotas \citep{kamada2015efficient}, and  controlled school choice due to considerations of gender or demographic balance \citep{ehlers2011school}.  
Our approach also handles problems that had not been analyzed via markets before, such as the well-known roommate problem from matching theory, and coalition formation problems. Our approach can be extended to accommodate participation constraints that arise from initial endowments.

The key idea in our proposal is to {\em price}  constraints. Think of two workers, Alice and Bob, in HZ's model. Jobs are in fixed supply. If Alice buys a probability share of job J, there is less left for Bob. Such ``pecuniary'' externalities are handled in pseudo-markets by pricing J. In our proposal, we think of the price of J as the price on the supply constraint for J: the constraint saying that the total demand for J cannot exceed its supply. HZ show that  pecuniary externalities can be correctly internalized by equilibrium prices, and that an efficient outcome results in equilibrium. In the present paper, we interpret all sorts of constraints, well beyond supply constraints, as giving rise to pecuniary externalities. In consequence, we use prices to align agents' choices with constraints and to find an efficient outcome. Under standard continuity, convexity, and monotonicity assumptions, we prove the existence of  pseudo-market equilibria. Every equilibrium outcome is (constrained) efficient.  When constraints do not single out any particular agent, every equilibrium outcome is fair. When agents start out with initial endowments, every equilibrium outcome can be approximately individually rational.\footnote{When agents' incomes equal the value of their endowments at equilibrium prices, HZ have shown that pseudo-market equilibria may not exist. We solve this nonexistence problem by mixing endogenous endowment values with an arbitrarily small exogenous budget.}

Specifically, we differ from the literature in that we do not take as primitive a formal description of constraints. In this sense, our work is orthogonal to the question of implementability of random assigments, which is the main concern of the influential work by \cite{budish2013designing} and the recent development by \cite{akbarpourapproximate}. We start from a set of feasible ex-post assignments, without specifying a formal model of constraints that the feasible assignments must satisfy. The set of random assignments that comply with the constraints is the convex hull of feasible ex-post assignments, and thus implementable by construction. Our primitive is this convex hull: a polytope in finite assignment problems.\footnote{Therefore, we can address any constraints that pin down a well-defined set of feasible ex-post assignments.} We proceed by using  linear inequalities to characterize the ``upper-right'' boundary of the convex hull. These linear inequalities include the standard supply constraints and additional constraints in specific applications. Each constraint is then priced. When Alice purchases one unit of good J, she will have to pay to the extent that her purchase affects other agents through different constraints. For example, if there is a ceiling constraint on how many units of J can go to a group of agents, the agents in the group will pay the price of the constraint when they buy units of J. If those agents are also involved in other constraints, then the final personalized prices they face can be different. But if two agents are always involved in the same constraints, they will face equal prices. Equal budgets then ensure that they will not envy each other in the resulting equilibrium outcome.

The idea may seem familiar from the role of shadow prices in optimization with constraints, but the familiarity is deceptive. Imagine using the dual variables (or Lagrange multipliers) associated with each constraint in order to decentralize an allocation that is constrained efficient. We run into two issues. One is that some constraints that impose lower bounds on consumed quantities would lead to negative prices. The other is that decentralizing a constrained efficient allocation would require transfers, as in the second welfare theorem.\footnote{We discuss the second welfare theorem without transfers developed by \cite{mirallesfoundations} in the related literature section (Section \ref{sec:relatedliterature}). Note that the outcome of the second welfare theorem may induce envy, even in a textbook economy with only supply constraints.} With endogenous transfers, one cannot ensure a fair outcome, or an individually rational outcome when there are endowments. Our approach, in contrast, can ensure fairness and individual rationality because our prices constitute market equilibria. By pricing only constraints on the ``upper-right'' boundary, we ensure that all prices in our approach are nonnegative (Sections~\ref{sec:preprocess} and~\ref{sec:mainresult}).  Yet we prove that the equilibrium outcome satisfies all constraints. As long as the constraints do not themselves induce unfairness by treating agents differently, we can obtain a fair outcome. Individual rationality can be ensured when agents have endowments (Section~\ref{sec:endowment}). Finally, observe that existing results on the computational complexity of finding HZ equilibria \citep{vazirani2020computational} imply that it is impossible to obtain our results through convex programming duality because HZ's model is a special case of ours. Moreover the complexity results have been extended to our setting with endowments (actually using the ideas introduced in our paper) by \cite{garg2020arrowdebreu}, 

We turn to a discussion of specific applications that motivate our approach, and where our results deliver new insights.

\subsection{Motivation}

\subsubsection{Matching jobs to workers.} HZ illustrated the use of pseudo-markets by way of assigning jobs to workers. Each worker is to receive at most one job, which we call a \df{unit demand constraint}. Each job is in unit supply, so the sum of probability shares of a particular job assigned to workers cannot exceed one, which we call a \df{unit supply constraint}. As we shall see, constraints that only involve an individual agent, such as unit demand constraints, have  no external effects and do not need to be priced. Constraints that involve multiple agents, such as supply constraints,  will be priced. When an agent purchases a good, she has to pay to the extent that the purchase impacts the supply constraints that are priced.

Importantly, there is one supply constraint for each good. The price that corresponds to the supply constraint for good $l$ is the familiar ``price of good $l$.'' Think of the constraint as capturing a \emph{pecuniary externality}. When Alice purchases good $l$, the supply constraint implies that there is less good $l$ available for Bob. By pricing supply constraints we ensure that Alice internalizes the effects that her purchases has on Bob. As we shall see, in problems with more complex constraint structures, we may not be able to ascribe a specific good to each specific price; but the logic of using prices to internalize the external effects induced by constraints extends.

In the jobs-to-workers application, priced constraints affect all agents in the same way. As a consequence, the prices are equal for all agents, and we can show the existence of a market equilibrium outcome that is both efficient and fair. 

Finally, the methodology in our paper extends the HZ approach to situations with multiple unit assignment; see Section~\ref{sec:combinatorialalloc} (and \cite{budish2012matching} for an overview).\footnote{In multi-unit assignment, the relevant set of constraints that need to be priced does not necessarily coincide with the set of items or the set of bundles. See Section~\ref{sec:combinatorialalloc} for a brief discussion.}

\subsubsection{Assigning doctors to hospital positions.} The problem of assigning doctors to hospitals is similar to the jobs-to-workers example, but with an important twist. Hospitals belong to different geographical regions, and the system seeks to ensure a minimum number of doctors per region \citep{kamada2015efficient}. So we have unit demand and supply constraints as before; but there is now a lower bound --- a \emph{floor constraint} --- on the number of doctors assigned to each region. Our solution turns these constraints into upper bounds (see Section~\ref{sec:floorconstraints} for details). 

There are then two kinds of non-individual constraints that must be priced: supply constraints and the modified floor constraints. When Alice buys into a popular hospital position, she causes a pecuniary externality on Bob, who may have to take a position in a less-demanded regional hospital. The price on the corresponding modified floor constraint ensures that she pays more for the popular hospital than if she were only facing the supply constraint. In an equilibrium, now, prices ensure that demand spills over into less attractive regions so as to meet the lower bound for each region.

In the doctors-to-hospitals application, again, all agents are treated in the same way by the priced constraints. In consequence, all agents face the same prices, and we obtain a market equilibrium that is efficient and fair. The application to doctor-hospital matching with regional constraints is discussed in Section~\ref{sec:HZwconstraints}.

\subsubsection{Roommates.} A set of college students need to pair up as roommates. Each student has a utility function defined over her possible roommates. We formulate the model as an assignment problem,  where we assign objects to agents by treating objects as copies of agents. Each student has two roles: one as agent seeking to match to an object, and one as object that can be matched to different agents.  In addition to the familiar unit demand and supply constraints, we must now impose a symmetry constraint. If agent Alice is matched to object Bob, then agent Bob must be matched to object Alice. The symmetry constraints involve more than one agent, and must therefore be priced. When Alice purchases some of the ``Bob good'' she is committing Bob to consume an equal amount of the ``Alice good.'' In our pseudo-market, this pecuniary external effect is internalized via prices. 

Our result delivers an efficient equilibrium in the market for roommates. Our finding is significant because it is well known that stable matchings may not exist in the model of roommates. Market equilibria capture a different notion of stability, one that is not game theoretic in nature, ensuring that agents are optimizing at the equilibrium outcome. 

The application to roommates is discussed in Section~\ref{sec:roomates}, where we also outline how pseudo-markets can be used in more general matching and coalition formation problems.

\subsubsection{Re-allocation of an existing assignment.} Consider a system that starts from a pre-existing assignment, for example assigning teachers to schools \citep{combe2018design}, offices to university faculty \citep{baccarayariv}, or upper-class students to dorms \citep{abdulkadirouglu1999house}. We seek to re-assign objects while ensuring each agent that they are not worse off than under the pre-existing allocation. The reason could be political. Re-assignment problems often need to overcome political economy obstacles stemming from an existing assignment. By ensuring that agents are not made worse off, the political problem is avoided. 

Our results allow agents to obtain market income from the value of their endowment at equilibrium prices. In consequence, we can not only achieve efficiency, but also satisfy the participation constraints implied by the presence of endowments.  It is crucial, as we show, that not all incomes are derived from endowments: some income must have the same external source as in HZ, since, as HZ showed, equilibria are not ensured to exist if all incomes equal market values of endowments.

Endowments and participation constraints are discussed in Section~\ref{sec:endowment}. 

\subsection{Related literature} Constrained resource allocation has received a lot of attention in recent years. The work by \cite*{kojima2018job},  \cite*{gulpesendorferzhang} and ours seems to be the first to look at constrained allocation by way of a market mechanism. The former two papers study the role of gross substitutes in a general model of discrete allocation. Despite a similar focus on markets and constraints, the results in our papers are very different; see Section~\ref{sec:relatedliterature} for more details. 
In studying constraints, we are motivated by the early work of \cite{budish2013designing}, \cite*{ehlers2011school} and \cite*{kamada2015efficient}. Our results apply to constraints well beyond those considered by these authors, and we differ substantially in methodology.  
Aside from how we deal with constraints, our approach to generating income from endowments is closely related to, but distinct from, \cite{mas1992equilibrium}, \cite{le2017}, and \cite{mclennan2018}. 

We provide a detailed discussion of the related literature in  Section~\ref{sec:relatedliterature}, once our results have been explained. We also provide a detailed comparison with other work on the pseudo-market mechanism in that section.

\section{The model}\label{sec:model}

\subsection{Notational conventions}

For vectors $x,y\in\Re^n$, $x\leq y$ means that $x_i\leq y_i$ for all $i=1,\dots, n$; $x < y$ means that $x\leq y$ and $x\neq y$; and $x\ll y$ means that $x_i < y_i$ for all $i=1,\dots, n$. The set of all $x\in \Re^n$ with $0\leq x$ is denoted by $\Re^n_+$, and the set of all $x\in \Re^n$ with $0\ll x$ is denoted by $\Re^n_{++}$. Inner products are denoted as $x\cdot y = \sum_i x_i y_i$. 

Let $X\subseteq \Re^n$ be convex. 
A function $u:X\rightarrow \Re$ is 
\begin{itemize}
\item \df{quasi-concave} if for any $ x,z\in X $ and $ \la\in [0,1] $, \[ \min \{u(z),u(x)\}\le
u(\la z+(1-\la)x). \] 
\item \df{semi-strictly quasi-concave} if it is quasi-concave, and for
	any $x,z\in X$ and $\la\in (0,1)$, $u(z)\neq u(x)$ implies that  \[\min \{u(z),u(x)\}<
	u(\la z+(1-\la)x).\footnote{See \cite{avriel2010generalized} for a discussion of semi-strictly quasi-concave functions and their applications to economics.}\]
 \item \df{concave} if, for any $x,z\in X$ and $\la\in (0,1)$,
 \[
 \la u(z)+(1-\la) u(x)\leq  u(\la z+(1-\la)x).
 \]
\item \df{expected utility} if there exists a vector $v\in\Re^n$ with $u(x)=v\cdot x$ for all $x\in X$. 
\item \df{strictly increasing} if $ x> x'$ implies that $ u(x)>u(x')$.
\end{itemize}

Given a set $A\subseteq \Re^n,$ let $\co(A)$ denote the \df{convex hull} of $A$ in $\Re^n$: the intersection of all convex sets that contain $A$.

A pair $(a,b)$, with $a\in \Re^n$ and $b\in\Re$, defines a \df{linear inequality} $a\cdot x\leq b$. We say that a linear inequality $(a,b)$ has \df{non-negative coefficients} if $a\geq 0$ and $b\geq 0$. Any linear inequality $(a,b)$ defines a (closed) \df{half-space} $\{x\in \Re^n:a\cdot x \leq b\}$.

A \df{polyhedron} in $\Re^n$ is a set that is the intersection of a finite number of closed half-spaces. A \df{polytope} in $\Re^n$ is a bounded polyhedron. Two special polytopes are the \df{simplex} in $\Re^n$:\[ \Delta^n = \{x\in \Re^n_+:\sum_{l=1}^n x_l=1 \}, 
\] and the \df{subsimplex}  \[ \Delta_{-}^n = \{x\in \Re^n_+:\sum_{l=1}^n x_l\leq 1 \}.
\] When $n$ is understood, we use the notation $\Delta$ and $\Delta_{-}$.

\subsection{The economy}\label{sec:theeconomy}
We first introduce a model without endowments. It is simpler, and goes a long way to capture the applications we have discussed. In Section~\ref{sec:endowment} we augment the model to include agents' endowments.

An \df{economy} is a tuple $\Gamma=(I,O,(Z_i,u_i)_{i\in I},(q_l)_{l\in O})$, where \begin{itemize}
    \item $I$ is a finite set of \df{agents}, with $\nagents=|I|$;
    \item $ O $ is a finite set of \df{objects}, with $\nobjects =\abs{O}$;
    \item $Z_i\subseteq \Re^\nobjects_+$ is the \df{consumption space} of $i\in I$;
    \item $u_i:Z_i\rightarrow \Re$ is the \df{utility function} of $i\in I$;
    \item $ q_l\in \Re_{++} $ is the amount of $ l\in O $.
\end{itemize}

In an economy, $\nagents=|I|$ is the number of agents, and $\nobjects=\abs{O}$ is the number of different \df{objects}. Each object $l\in O$ is available in quantity $q_l$. 

An \df{assignment} in $\Gamma$ is a vector \[ x = (x_{i,l})_{i\in I, l\in O} \text{ with } x_{i}\in Z_i,
\] where $x_{i,l}$ is the amount of object $l$ received by agent $i$.  Let $\A$ denote the set of all assignments in $\Gamma$.

In discrete allocation problems we often interpret assignments as probabilistic allocations: see Section~\ref{sec:discreteallocation}. In this case, $x_{i,l}$ is the probability that agent $i$ receives a copy of an object $l$. 

For now we restrict attention to $ Z_i = \Re^{\nobjects}_+ $, but agents' consumption spaces will be restricted further as we introduce constraints.

\subsection{Constraints}\label{sec:constraints}

A \df{constrained allocation problem} is a pair $(\Gamma,\C)$ in which $\Gamma$ is an economy and $\C$ is a subset of the assignments in $\Gamma.$ The assignments in $\C$ constitute the assignments that satisfy the constraints: the \df{feasible} assignments in $(\Gamma,\C)$. 

Observe that the set of feasible assignments is a primitive of our model. Instead of starting from an explicit description of how assignments are constrained, we work directly with the set $\C$ of feasible assignments. In Section~\ref{sec:constraintstructures} we show how to apply our results to a model with explicit constraint structures. Throughout the paper we require that $ \C $ be a polytope.

Our model applies to environments with infinitely divisible objects. In Section~\ref{sec:endowment}, for example, we obtain the textbook model of an exchange economy as a special case. Most market design applications, however, require indivisible objects. We proceed to introduce some language that is pertinent to the indivisible case.

\subsection{Special case: discrete allocation}\label{sec:discreteallocation}

In many market design applications, objects are indivisible, and randomization over deterministic assignments is used to ensure fairness. In these applications, we say that an assignment $ x $ is \df{deterministic} if every $ x_{i,l} $ is an integer. When an assignment is not deterministic we call it a random assignment. 

Constraints are often imposed as linear inequalities on deterministic assignments. For example, the usual \df{unit-demand constraints} require that $\sum_{l\in O} x_{i,l}\le 1$ for all $i\in I$, and the \df{supply constraints} require that $\sum_{i\in I} x_{i,l} \le  q_l$ for all $l\in O$. Of course, constraints can also take other forms. A deterministic assignment is feasible if it satisfies all constraints. A random assignment is \df{feasible} if it belongs to the convex hull of feasible deterministic assignments. Thus we obtain $\C$ from the set of feasible deterministic assignments. Observe that $\C$ is a polytope, as the number of feasible deterministic assignments is finite.

Under unit demand and supply constraints, the Birkhoff-von Neumann theorem and its generalizations guarantee that every random assignment is a randomization over deterministic assignments. \cite{budish2013designing} generalize this theorem by characterizing constraint structures that ensure the implementation of feasible random assignments (see Section~\ref{sec:constraintstructures}). By taking $\C$ as a primitive, we circumvent the implementation issue.

\subsection{Pre-processing of constraints}\label{sec:preprocess}
 
 Our approach involves pricing constraints, but not all constraints get a price. For example, in HZ, unit demand constraints are not priced; only supply constraints get a price. Here we proceed with a general constrained allocation problem $(\Gamma,\C)$, and ``pre-process'' $\C$ so as to obtain the constraints that have to be assigned a price. 
 
 Recall that $\C$ is a polytope. Define the \textbf{lower contour set} of $\C$ to be
\[\lcs(\C) =\{x\in \Re^{\nagents \nobjects}_+ :\exists x'\in \C \text{ such that } x\le x'\}.
\]

\begin{lemma}\label{lem:nonnegativelinearineq}
There exists a finite set $\W$ of linear inequalities with non-negative coefficients such that \[
\lcs(\C) = \bigcap_{(a,b)\in \W} \{x\in \Re^{\nobjects \nagents}_+ : a\cdot x\leq b \}.\footnote{Lemma~\ref{lem:nonnegativelinearineq} is used by \cite{balbuzanov2019constrained} to define a generalization of the probabilistic serial mechanism that accommodates constraints.}\]
\end{lemma}
\begin{proof}
  Consider \[D=\{x'\in\Re^{\nagents \nobjects} : x'\leq x \text{ for some } x\in \C\}\] and note that $\lcs(\C)=D\cap \Re^{\nagents \nobjects}_+$. Write $D$ as $\C-\Re^{\nagents \nobjects}_+$; thus, since $\C$ is a polytope, $D$ is finitely generated. Then by Theorem 19.1 in \cite{rockafellar1970convex} $D$ is polyhedral, and therefore the intersection of finitely many halfspaces. Let $\W$ be the set of linear inequalities $(a,b)$ defining this collection of halfspaces, so for each $(a,b)\in \W$ we have the halfspace $\{x'\in\Re^{\nagents \nobjects}:a\cdot x'\leq b\}$.  Since for each $i$ and $l$  there is $x'\in D$ with arbitrarily small $x'_{i,l}$, we must have $a\geq 0$. If $\C=\{0\}$ we may take $b=0$. If there is $x\in\C$ with $x>0$ then $b\geq 0$.   Hence $\W$ defines a finite collection of linear inequalities with non-negative coefficients.

 To finish the proof, note that if $z\in D\setminus \lcs(\C)$ then $z\notin \Re^{\nagents \nobjects}_+$.
\end{proof}

For any $ c=(a,b)\in \W$, define the support of $c$ to be
\[\supp(c)=\{(i,l)\in I\times O:a_{i,l}>0\}.\] 
Let $ a_i=(a_{i,l})_{l\in O} $ be the vector of coefficients relevant to $ i $ in $ (a,b) $.

There are two types of inequalities $c=(a,b)$ in $\W$: those with $b=0$ and those with $b>0$. If $b=0$, then for any $x\in \C$ we must have $ x_{i,l}=0 $ for all $(i,l)\in \supp(c) $. We can, without loss of generality, assume that there is exactly one such inequality; because two inequalities $(a,0),(a',0)\in \W$ can be substituted by $((\max\{a_{i,l},a'_{i,l}\}),0)$, and if there is no inequality with $b=0$ then we can include the trivial inequality $(0,0)$ in $\W$. Thus we can let $(a^0,0)\in \W$ be the unique inequality with $b=0$. When $(a^0,0)$ is nontrivial, it forbids some agents from consuming certain objects. So we say that $l$ is a \df{forbidden object} for agent $i$ if $a^0_{i,l}>0$. 

Among the remaining inequalities $ \W\backslash \{(a^0,0) \} $, we say $(a,b) $ is an \df{individual constraint} for agent $ i $ if for all $ j\neq i $ and $ l\in O $, $ a_{j,l}=0 $. In words, $ (a,b) $ only restricts $ i $'s consumption. Let $ \W^i $ denote the set of all individual constraints for $ i $.  We use $(a^0,0)$ and individual constraints to refine $ i $'s consumption space. Let $\X_i$ be the set of vectors $x_i\in Z_i$ such that $x_{i,l}=0$ if $l$ is a forbidden object for $i$ and $ x_i $ satisfies all of $ i $'s individual constraints. That is, 
\[
\X_i =\{x_i\in Z_i :  a^0_i\cdot x_i \leq 0 \text{ and }a_i\cdot x_i\le b \text{ for all }(a,b)\in \W^i\}.
\]

Let $\W^* = \W\backslash \big(\{(a^0,0) \}\bigcup \cup_{i\in I}\W^i\big)$ collect all constraints that involve more than one agent, including, for example, any supply constraints. The elements of $\W^*$ will be ``priced.'' By pricing these constraints we seek to ensure that one agent's pecuniary externality on others, imposed via the constraints present in $\C$, are internalized.

\begin{remark}
    It is important to set aside and not price individual constraints because it will allow our fairness conclusion to be stronger. For example, in the special case of the HZ model, our approach avoids personalized prices altogether; in consequence we obtain (unqualified) envy-freeness.\footnote{That said, there is no additional difficulty in proving our main result if we were to include prices for the constraints in each $\W^i$.}

Fairness considerations aside, recall our motivation in terms of pricing externalities. Only the inequalities in $\W^*$ generate externalities, so it makes sense to price these, and not individual constraints.
\end{remark}



\subsection{Normative properties}\label{sec:normativeprop}
Given a constrained allocation problem  $(\Gamma,\C)$, we analyze constrained versions of efficiency and fairness:  the efficiency and fairness properties that can be achieved subject to how assignments are constrained. 

A feasible assignment $x\in \C$ is \df{weakly $\C$-constrained Pareto efficient} if there is no feasible assignment $y\in \C$ such that $ u_i(y_i)> u_i(x_i) $ for all $ i $. And $ x\in \C $ is \df{$\C$-constrained Pareto efficient} if there is no feasible assignment $ y\in \C $ such that $ u_i(y_i)\geq u_i(x_i) $ for all $ i $ with a strict inequality for at least one agent.

Fairness rules out envy among agents who are treated symmetrically by the primitive constraints. We say that two agents $i$ and $j$ are of \df{equal type} if $\X_i=\X_j$ and, for all $(a,b)\in \W^*$, $ a_i=a_j $. An agent $i$ \df{envies} another agent $j$ at an assignment $x$ if $u_i(x_j)>u_i(x_i)$. An assignment $x\in \C$ is \df{envy-free} if no agent envies another agent at $x$, and \df{equal-type envy-free} if no agent envies another agent of equal type at $x$. 

\subsection{Equilibrium} \label{sec:eqmdefn}
For each constraint $c=(a,b)\in \W^*$, we introduce a price $ p_c$. When agent $i$ purchases $x_{i,l}$, she affects other agents' purchases through the role of $a_{i,l}$ in constraint $c$. Prices are meant to internalize such effects, just as the price of good $l$  classically internalizes the effect that $i$ has on other agents through the supply constraint for good $l$. Given a price vector $ p=(p_c)_{c\in \W^*}\in \Re^{\W^*} $, the personalized price vector faced by any agent $ i$ is defined to be $ p_i=(p_{i,l})_{l=1}^\nobjects $ such that
\[
p_{i,l}=\sum_{(a,b)\in \W^*}a_{i,l}p_{(a,b)}.
\] 

\begin{remark}
    If agents $i$ and $j$ are of equal type, then $p_i=p_j$. Thus prices are only personalized to the extent that constraints are personalized. We present several applications where all agents face the same prices. 
\end{remark}

A pair $ (x^*,p^*) $ is a \df{pseudo-market equilibrium} for $(\Gamma,\C)$ if
\begin{enumerate}
	\item $ x^*_i\in \arg\max_{x_i\in \mathcal{X}_i}\{u_i(x_i):p^*_i\cdot x_i\leq 1\}$;
	\item  $x^*\in \C$;
	\item\label{it:complemenaryslackness} $ c=(a,b)\in \W^* $, $ a\cdot x^*<b$ implies that $ p^*_c=0 $.    
\end{enumerate}

\begin{remark} Condition~(\ref{it:complemenaryslackness}) may seem unfamiliar. It is a ``complementary slackness'' property that rules out trivial equilibria. We follow \cite{budish2013designing} in imposing (\ref{it:complemenaryslackness}) directly on the definition of equilibrium.
\end{remark}

\section{Main theorem}\label{sec:mainresult}

For constrained allocation problems without endowments, our main result is:

\begin{theorem}\label{thm:existence:noend}
	Suppose that agents' utility functions are continuous, quasi-concave and strictly increasing. 
	\begin{itemize}
		\item There exists a pseudo-market equilibrium $(x^*,p^*)$ in which $ x^* $ is weakly $\C$-constrained Pareto efficient.
		\item If agents' utility functions are semi-strictly quasi-concave, there exists a pseudo-market equilibrium $(x^*,p^*)$ in which $ x^* $ is $\C$-constrained Pareto efficient.
		\item Every pseudo-market equilibrium assignment is equal-type envy-free.
	\end{itemize}
\end{theorem}

Theorem~\ref{thm:existence:noend} is implied by our more general result, Theorem \ref{thm:existence} in Section \ref{sec:endowment}. We should emphasize that the first welfare theorem does not hold in our model: one can exhibit examples of Pareto inefficient pesudo-market equilibria, and even of Pareto-ranked equilibrium allocations. Crucial to Theorem~\ref{thm:existence:noend} is the cheapest bundle property:  $(x,p)$ satisfies the \df{cheapest-bundle property} if, for each $i$, $x_i$ minimizes expenditure $p_i\cdot z_i$ among all the $z_i\in \X_i$ for which $u_i(z_i)=u_i(x_i)$. The cheapest bundle property, and its role in obtaining efficiency, was already established by HZ. We show the existence of a pseudo-market equilibrium with the cheapest-bundle property, which in consequence is  $\C$-constrained Pareto efficient. 

The assumption that $u_i$ be strictly increasing is also worth commenting on. In the special case of discrete objects that are allocated probabilistically (as in HZ), it amounts to assuming that there is an outside option: A fixed good (in enough supply) that is obtained with probability $1-\sum_l x_{i,l}$. Monotonicity then means that all goods in $O$ are strictly preferred to the outside option. For floor constraints, monotonicity means that all objects are desirable to the extent that pricing  upper bound constraints will cause demand spill over, and thus ensure the satisfaction of floor constraints.\footnote{Our approach will not work when floor constraints involve consumption of goods that are ranked below the outside option in agents' preferences, but in such situations floor constraints are probably impossible to satisfy.}

\section{Application: Discrete object allocation under constraints}\label{sec:HZwconstraints}

Our first application is to the problem of assigning indivisible objects, as in Section~\ref{sec:discreteallocation}. We seek fair and efficient random assignments subject to constraints.  Each object $l\in O$ is available in fixed integer supply, and each agent demands at most one copy of any object. 

We explicitly describe constraints by way of the constraint structures introduced by \cite{budish2013designing}. Many constraints in real-life allocation problems can be described through such structures.  

\subsection{Constraint structures}\label{sec:constraintstructures} 

A \df{constraint} is denoted by a tuple $(S,(\underline{q}_S, \overline{q}_S))$, where $S \subset I\times O $ is a subset of agent-object pairs, and $q_S=(\underline{q}_S,\overline{q}_S)$ is a pair of non-negative integers with $\underline{q}_S\le \overline{q}_S$.  The integers $ \underline{q}_S $ and $ \overline{q}_S $ are respectively called \df{floor}  and \df{ceiling} quotas. 

An assignment $x$ satisfies $(S,(\underline{q}_S, \overline{q}_S))$ if
\begin{equation}\label{eq:inequalityform}
	\underline{q}_S\le \sum_{(i,l)\in S}x_{i,l}\le \overline{q}_S.
\end{equation}

For any $ i\in I $ and $ l\in O $, a \df{singleton constraint} $(S,q_S)$ is such that $ S=\{(i,l)\} $ and $ q_S=(0,1) $. This means that $ i $ can obtain at most one copy of $ l $. For any $ i\in I $, a \df{row constraint} $(S,q_S)$ is such that $ S=\{i\}\times O $ and $ q_S=(0,q_i) $ where $ q_i\in \mathbb{N} $. This means that $ i $ obtains at most $ q_i $ objects. Unit demand is an example of a row constraint. For any $ l\in O $, a \df{column constraint} $(S,q_S)$ is such that $ S=I\times \{l\} $ and $ q_S=(0,q_l) $. This means that at most $ q_l $ copies of $ l $ can be assigned. Supply constraints are examples of column constraints. 

A \df{constraint structure} $\H$ is a collection of constraints. We assume that $\H$ contains the singleton constraints for all agent-object pairs, the row constrains for all agents and the column constraints for all objects.  When $\H$ satisfies these assumptions, we say that it is \df{allocative}.

The set of assignments that satisfy the constraints in $\H$ is defined to be
 \[ 
\C = \{x\in \Re^{\nagents \nobjects}_+: \underline{q}_S\le \sum_{(i,l)\in S}x_{i,l}\le \overline{q}_S \text{ for all } (S,q_S)\in \H \}.
\] 


Note that $\C$ could be different from the convex hull of deterministic assignments satisfying the constraints in $ \H $. When they coincide, we say that $\C$ is \df{implementable}. \cite{budish2013designing} prove that a sufficient and necessary condition for $ \C $ to be implementable for all possible quotas is that $ \H $ is a \df{bihierarchy}. Formally, a constraint structure $ \mathcal{H} $ is a \df{hierarchy} if for every distinct $ S $ and $ S' $ in $ \mathcal{H} $, either $ S\subset S' $, or $ S'\subset S $, or $ S\cap S'=\emptyset $. $ \mathcal{H} $ is a bihierarchy if there exist two hierarchies $ \mathcal{H}_1 $ and $ \mathcal{H}_2 $ such that $ \mathcal{H}=\mathcal{H}_1\cup \mathcal{H}_2 $ and $ \mathcal{H}_1\cap \mathcal{H}_2=\emptyset $. $ \H_1 $ is the set of sub-row, row, and sup-row constraints, while $ \H_2 $ is the set of sub-column, column and sup-column constraints.\footnote{A constraint $ (S,q_S) $ is a sub-row constraint if $ S=\{i\}\times O' $ for some $ i\in I $ and $ O'\subset O $, and it is a sup-row constraint if $ S=I'\times O $ for some $ I'\subset I $. Sub-column and sup-column constraints are similarly defined.} When $ \mathcal{H} $ is a bihierarchy, $ \C $ is the set of feasible assignments that we take as a primitive in our model. Then we can apply our method directly to $ \C $.  In the rest of this section, we discuss applications with bihierarchy constraints. Note, however, that our approach is applicable to cases where $\H$ is not a bihierarchy. In such cases, one first needs to describe the convex hull of deterministic assignments that satisfy the constraints in $\H$ (see our application to the roommate problem in Section \ref{sec:roomates}).

We consider two applications of bihierarchy constraints. In the first application, all floor quotas are zero. Then we can directly price the constraints in $ \H $. It should be clear that singleton and row constraints do not need to be priced. Column constraints involve more than one agent, and thus generate pecuniary externalities that must be internalized through prices.  In the second application, there are nontrivial floor constraints. To characterize $ \lcs(\C) $, we will then derive a new set of ceiling constraints implied by the ceiling and floor constraints present in $\H$. We discuss two concrete examples to show how the new ceiling constraints are derived. The general model can be treated much like our examples.

\subsection{No floor constraints}\label{sec:nofloor}

Suppose that $\H$ is such that all floor quotas are zero. Then $ \C=\lcs(\C) $, and we can directly price all non-individual constraints in $ \H $: a set that we denote by  $ \H^* $. Here individual constraints consist of singleton, sub-row, and row constraints. The set $\W^*$ is \[ 
\{ (\one_{S},\overline q_S) \in \Re_+^{\nobjects \nagents} \times \Re_+ : (S,(0,\overline q_S))\in \H^* \}.\footnote{By $\one_S$ we denote the indicator vector of the set $S$.}
\] Under our assumptions on utilities, an efficient pseudo-market equilibrium exists with prices for each constraint in $\H^*$. It is interesting to discuss the fairness properties of such equilibria.

Two agents $i$ and $j$ are of \df{equal type} if $\X_i=\X_j$ and, for all $S\in \H^*$ and $l\in O$, $(i,l)\in S$ if and only if $(j,l)\in S$. We say that $ \H $ is \df{anonymous} if every two agents are of equal type. If $ \H $ is anonymous, every constraint in $\H^*$ must be a column or sup-column constraint.  Under anonymous constraints, every pseudo-market equilibrium is envy-free: the strongest possible fairness property that we can obtain.

An example with anonymous constraints is the Japanese medical residency match with regional caps studied by \cite{kamada2015efficient}. Suppose that agents are doctors and objects are hospital positions. Each constraint in $\H^*$ takes the form \[
\big(I\times O',(0,\bar q_{O'})\big)
\] 
where $ O'\subseteq O $ is the set of hospitals in a geographic region (a city or a prefecture). Here $ \bar q_{O'}  $ is the regional cap used to control the maximum number of doctors that the region $ O' $ can employ. A collection of such constraints is anonymous because each constraint does not distinguish among the identities of individual doctors.

\subsection{Floor constraints}\label{sec:floorconstraints}

Floor constraints are common in applications, but difficult to deal with theoretically. We discuss two examples. The \textbf{first} is the Japanese medical residency match mentioned above. By introducing regional caps to restrict the number of doctors assigned to urban hospitals, the Japanese government wants to increase the number of doctors assigned to rural hospitals. But the government's ideal distribution of doctors can be described by a collection of constraints with \emph{both} floor and ceiling quotas. (The regional caps described above can be motivated through our approach to characterizing $\lcs(\C)$.)

Specifically, the hospitals $ O $ are located in $ K $ disjoint regions. Accordingly, there is a partition of hospitals $ O=R_1\cup R_2\cup \cdots \cup R_K $ such that every $ R_k $ is the set of hospitals in a region. So we simply refer to every $ R_k $ as a region. For each region $ R_k $, there is a constraint
\[
\underline{q}_{R_k}\le \sum_{l\in R_k}\sum_{i\in I} x_{i,l}\le \overline{q}_{R_k}.
\]

Assume that there are enough hospital positions to assign each doctor a position. We can always add null hospitals when that is not the case. Also, assume that there are enough doctors to meet all floor constraints. Below we derive the inequalities in $ \W $ and show that they are anonymous. Our theorem applies to deliver a pseudo-market equilibrium that satisfies the constraints, and is efficient and envy-free. 

Let $ \mathcal{R}=\{R_1,R_2,\ldots,R_K\} $ denote the set of regions. For each $ \ell\in \{1,2,\ldots,K\} $, let $ \mathcal{R}_\ell $ be the collection of sets that are the union of $ \ell $ distinct regions. That is,
\[
\mathcal{R}_\ell=\{R_{k_1}\cup R_{k_2}\cup \cdots \cup R_{k_\ell}:\{k_1,k_2,\ldots,k_\ell\}\subset \{1,2,\ldots,K\} \}.
\] In particular, $ \mathcal{R}_1=\mathcal{R} $.

Consider the following inequalities
\begin{equation}\label{eq:hospital}
	\begin{cases}
	0\le \sum_{l\in O}x_{i,l}\le 1 &\text{ for all }i\in I,\\
	0\le \sum_{i\in I} x_{i,l} \le q_l &\text{ for all }l\in O,\\
	0\le \sum_{i\in I, l\in R}x_{i,l} \le \overline{q}_R &\text{ for all }\ell \in \{1,\ldots,K\} \text{ and }R\in \mathcal{R}_\ell,
	\end{cases}
\end{equation}
where $ \overline{q}_R $ is (re)defined according to the following procedure:
\begin{itemize}
	\item For every $ R\in \mathcal{R}_1 $, redefine its ceiling quota to be
	\[
	\overline{q}_R= \min\big\{\overline{q}_R,N-\sum_{R'\in\mathcal{R}\backslash \{R\}}\underline{q}_{R'}\big\}.
	\]
	Note that $ \overline{q}_R\ge \underline{q}_R $ because $ N\ge \sum_{R'\in\mathcal{R}}\underline{q}_{R'} $, and $ \overline{q}_R $ is weakly smaller than the original ceiling quota.
	
	\item For every $ R=R_{k_1}\cup R_{k_2} \in \mathcal{R}_2 $, define its ceiling quota to be
	\[
	\overline{q}_R= \min\big\{\overline{q}_{R_{k_1}}+\overline{q}_{R_{k_2}},N-\sum_{R'\in \mathcal{R}\backslash \{R_{k_1},R_{k_2}\}}\underline{q}_{R'}\big\}.
	\]
    \item In general, for every $ R =R_{k_1}\cup R_{k_2}\cup \cdots \cup R_{k_\ell} \in \mathcal{R}_\ell$, define its ceiling quota to be
	\[
	\overline{q}_R= \min\big\{\overline{q}_{R\backslash \{R_{k_x}\}}+\overline{q}_{R_{k_x}} \text{ for every }x\in \{1,2,\ldots,\ell\},N-\sum_{R'\in\mathcal{R}\backslash \{R_{k_1},\ldots,R_{k_\ell}\}}\underline{q}_{R'}\big\}.
	\]	
\end{itemize}

We prove that $ \lcs(\C) $ is characterized by the inequalities in (\ref{eq:hospital}).

\begin{proposition}\label{prop:hospital}
	$ \lcs(\C)=\{ x\in \Re^{NL}_+:x \text{ satifsies the inequalities in (\ref{eq:hospital})}\} $.
\end{proposition}
\begin{proof}
	We denote by $ A $ the set characterized by the inequalities in (\ref{eq:hospital}). It is easy to see that $ A=\lcs(A) $. By the procedure to define $ \overline{q}_R $, all elements of $ \C $ satisfy (\ref{eq:hospital}). So $ \C\subset A $ and thus $ \lcs(\C)\subset A $. 	
	To prove that $ A\subset  \lcs(\C)$, we first prove a claim.
	
	\begin{claim*}\label{claim:distributional}
		For every $ \ell\in \{2,\ldots,K\} $, every $ R=R_{k_1}\cup R_{k_2}\cup \cdots \cup R_{k_\ell}\in \mathcal{R}_\ell $, and every $ x\in \{1,\ldots,\ell\} $,
		\begin{center}
			$ \overline{q}_R\ge \underline{q}_{R_{k_x}}+\overline{q}_{R\backslash \{R_{k_x}\}} $.
		\end{center} 
	\end{claim*}
	\noindent \df{Proof of the claim.}
		Base case $ \ell=2 $: For every $ R=R_{k_1}\cup R_{k_2} \in \mathcal{R}_2 $, if $ \overline{q}_R= \overline{q}_{R_{k_1}}+\overline{q}_{R_{k_2}}$, then the claim holds obviously. Otherwise, $ \overline{q}_R= N-\sum_{R'\in\mathcal{R}\backslash \{R_{k_1},R_{k_2}\}}\underline{q}_{R'} $. By definition, $ N-\sum_{R'\in\mathcal{R}\backslash \{R_{k_1}\}}\underline{q}_{R'}\ge \overline{q}_{R_{k_1}} $. So
		$ \overline{q}_R= N-\sum_{R'\in\mathcal{R}\backslash \{R_{k_1},R_{k_2}\}}\underline{q}_{R'}\ge \overline{q}_{R_{k_1}}+\sum_{R'\in\mathcal{R}\backslash \{R_{k_1}\}}\underline{q}_{R'}-\sum_{R'\in\mathcal{R}\backslash \{R_{k_1},R_{k_2}\}}\underline{q}_{R'}= \overline{q}_{R_{k_1}}+\underline{q}_{R_{k_2}}$. Similarly, we prove that $ \overline{q}_R\ge \underline{q}_{R_{k_1}}+\overline{q}_{R_{k_2}} $.
		
		Induction step: Suppose the claim is true for $ 1,2,\ldots,\ell $. Then we prove that it is also true for $ \ell+1 $. For any $ R =R_{k_1}\cup R_{k_2}\cup \cdots \cup R_{k_{\ell+1}} \in \mathcal{R}_{\ell+1}$, if $ \overline{q}_R= \overline{q}_{R\backslash \{R_{k_x}\}}+\overline{q}_{R_{k_x}}$ for some $ x\in \{1,2,\ldots,\ell+1\} $, then it is obvious that $ \overline{q}_R\ge \overline{q}_{R\backslash \{R_{k_x}\}}+\underline{q}_{R_{k_x}} $. By the induction assumption, for every $ y\neq x $, $ \overline{q}_{R\backslash \{R_{k_x}\}}\ge \underline{q}_{R_{k_y}}+\overline{q}_{R\backslash \{R_{k_x},R_{k_y}\}}$. So $ \overline{q}_R\ge \underline{q}_{R_{k_y}}+\overline{q}_{R\backslash \{R_{k_x},R_{k_y}\}}+ \overline{q}_{R_{k_x}}\ge \underline{q}_{R_{k_y}}+\overline{q}_{R\backslash \{R_{k_y}\}}$. The claim is proved.
		
		Otherwise,  $ \overline{q}_R= N-\sum_{R'\in\mathcal{R}\backslash \{R_{k_1},\ldots,R_{k_{\ell+1}}\}}\underline{q}_{R'}$. By definition, for every $ x $, $ N-\sum_{R'\in\mathcal{R}\backslash \{R_{k_1},\ldots,R_{k_{\ell+1}}\}\cup \{R_{k_x}\}}\underline{q}_{R'}\ge \overline{q}_{R\backslash \{R_{k_x}\}}$. So $ \overline{q}_R\ge \overline{q}_{R\backslash \{R_{k_x}\}}+\underline{q}_{R_{k_x}} $.
		
		By induction, we complete the proof of the claim. 
		
		\medskip
	
	Define $ A'=\{x\in A: \nexists x'\in A \text{ such that }x<x'\} $. It is clear that $ A=\lcs(A') $. We prove that $ A'\subset \C $. Suppose there exists $ x\in A' $ such that $ x\notin \C $. Because $ x $ satisfies all original ceiling constraints that define $ \C $, $ x $ must violate the floor constraint of some $ R_k $. That is, $ \sum_{i\in I, l\in R_k}x_{i,l}<\underline{q}_{R_k} $. Then there must exist some doctor $ i $ such that $ \sum_{l\in O}x_{i,l}<1 $, since otherwise $ \sum_{i\in I,l\in O\backslash \{R_k\}}x_{i,l}=N-\sum_{i\in I, l\in R_k}x_{i,l}>N-\underline{q}_{R_k}\ge \overline{q}_{O\backslash \{R_k\}} $, which contradicts the assumption that $ x\in A'\subset A $. Because $ \underline{q}_{R_k}\le \sum_{l\in R_k} q_l $, there must exist $ l\in R_k $ such that $ \sum_{i\in I}x_{i,l}<q_l $. Now consider a new assignment $ x' $ such that $ x'_{i,l}=x_{i,l}+\epsilon $ where $0< \epsilon<\min \{1-\sum_{l\in O}x_{i,l},q_l-\sum_{i\in I}x_{i,l},\underline{q}_{R_k}-\sum_{i\in I, l\in R_k}x_{i,l}\} $, and $ x' $ coincides with $ x $ in the other cells. So $ x<x' $. Below we prove that $ x'\in A $, which contradicts the assumption that $ x\in A' $.
	
	Suppose towards a contradiction that $ x'\notin A $. Let $ \ell>1 $ be the smallest index such that there exists $ R\in \mathcal{R}_\ell $ with $ \sum_{i\in I, l\in R}x'_{i,l}> \overline{q}_R $. It is clear that $ R_k\subset R $. By Claim, $ \overline{q}_R\ge \underline{q}_{R_k}+\overline{q}_{R\backslash R_k} $. So
	\[
	\sum_{i\in I, l\in R}x'_{i,l}> \underline{q}_{R_k}+\overline{q}_{R\backslash R_k}.
	\]Because $ \epsilon $ is chosen such that $ \sum_{i\in I, l\in R_k}x'_{i,l}<\underline{q}_{R_k} $. So
	\[
	\sum_{i\in I, l\in R\backslash R_k}x'_{i,l}>\overline{q}_{R\backslash R_k}.
	\]
	But it means that $ \sum_{i\in I, l\in R\backslash R_k}x_{i,l}>\overline{q}_{R\backslash R_k} $, which contradicts $ x\in A $. So $ x'\in A $.
\end{proof}
Besides unit demand constraints, the other inequalities in (\ref{eq:hospital}) do not distinguish among the identities of doctors. So every pseudo-market equilibrium is envy-free.

The \textbf{second} application we discuss is controlled school choice. When implementing school choice, a consideration for many school districts is demographic diversity. We present a model in which the students $ I $ are simply classified into minorities $ I^m $ and majorities $ I^M $.\footnote{Our arguments extend to more than two types.} Let the number of minorities be $ N^m $ and the number of majorities be $ N^M $. Each school $ l $ has a pair of quotas $ (\overline{q}^m_l,\underline{q}^m_l) $ for minorities, and a pair of quotas $ (\overline{q}^M_l,\underline{q}^M_l) $ for majorities. So aside from supply constraints, each school $ l $ has the constraints
\begin{align*}
	&\underline{q}^m_l\le \sum_{i\in I^m}x_{i,l}\le \overline{q}^m_l,\\
	&\underline{q}^M_l\le \sum_{i\in I^M}x_{i,l}\le \overline{q}^M_l.
\end{align*}
Of course, we assume that $ \underline{q}^m_l+\underline{q}^M_l\le q_l $.

The inequalities to characterize $ \lcs(\C) $ can be derived similarly to how we dealt with regional hospitals above. The only difference is that we need to take into account  the interaction between the quotas for the two student types within each school. After that, we can deal with the assignments for two types separately. 

Formally, consider the following inequalities
\begin{equation}\label{eq:controlled}
\begin{cases}
0\le \sum_{l\in O}x_{i,l}\le 1 &\text{ for all }i\in I,\\
0\le \sum_{i\in I} x_{i,l} \le q_l &\text{ for all }l\in O,\\
0\le \sum_{i\in I^m, l\in O'}x_{i,l} \le \overline{q}^m_{O'} &\text{ for all nonempty }O'\subset O,\\
0\le \sum_{i\in I^M, l\in O'}x_{i,l} \le \overline{q}^M_{O'} &\text{ for all nonempty }O'\subset O,
\end{cases}
\end{equation} 
where $ \overline{q}^m_{O'} $ and $ \overline{q}^M_{O'} $ are (re)defined as follows:
\begin{itemize}
	\item For every $ l\in O $, redefine the ceiling quotas to be
	\begin{align*}
	&\overline{q}^m_l= \min\{\overline{q}^m_l,q_l-\underline{q}^M_l,N^m-\sum_{l'\in O\backslash \{l\}}\underline{q}^m_{l'}\},\\
	&\overline{q}^M_l= \min\{\overline{q}^M_l,q_l-\underline{q}^m_l,N^M-\sum_{l'\in O\backslash \{l\}}\underline{q}^M_{l'}\}.
	\end{align*}
	
	\item For every non-singleton $ O'\subset O$, (inductively) define the ceiling quotas to be
	\begin{align*}
	&\overline{q}^m_{O'}= \min\{\overline{q}^m_{O'\backslash \{l\}}+\overline{q}^m_{l} \text{ for every }l\in O',N^m-\sum_{l'\in O\backslash O'}\underline{q}^m_{l'}\},\\
	&\overline{q}^M_{O'}= \min\{\overline{q}^M_{O'\backslash \{l\}}+\overline{q}^M_{l} \text{ for every }l\in O',N^M-\sum_{l'\in O\backslash O'}\underline{q}^M_{l'}\}.
	\end{align*}
\end{itemize}

Similarly as before, $ \lcs(\C) $ is characterized by the inequalities in (\ref{eq:controlled}).

\begin{proposition}\label{prop:controlled}
	$ \lcs(\C)=\{ x\in \Re^{NL}_+:x \text{ satifsies the inequalities in (\ref{eq:controlled})}\} $.
\end{proposition}

The proof of Proposition~\ref{prop:controlled} is similar to that of Proposition~\ref{prop:hospital} and thus omitted. 

Note that besides unit demand constraints, the other inequalities in (\ref{eq:controlled}) do not distinguish among the identities of the students of each type. So every pseudo-market equilibrium is envy-free among the students of each type. That is, minority students will not envy other minority students, and majority students will not envy other majority students. 

\section{A market for roommates, and other problems}\label{sec:roomates} Our model accommodates very general assignment problems with constraints, including models with non-implementable constraints. We discuss coalition formation problem as an illustration of the power of our approach. 

First we consider the roommate model, arguably the best-known example in matching theory where game-theoretic stability solutions fail to exist. As a corollary of our main theorem, we obtain the existence of efficient pseudo-market equilibrium assignments. Equilibria embody a form of stability: optimizing agents do not want to change their behavior in the market. In this sense, our results offer a possible way out of the non-existence of stable matchings.

Consider a set of agents $I$ that constitute the potential roommates or partners. Let $O$ be a copy of $I$, so $\nagents=\nobjects$, and think of $i\in O$ as the alter ego of agent $i\in I$. If $x$ is an assignment, interpret $x_{i,j}=1$ as agents $i$ and $j$ forming a partnership, or becoming roommates. When $i$ is alone without a roommate, we have $x_{i,i}=1$.  In consequence, we restrict attention to assignments $x$ where $x_{i,j} = x_{j,i}$, meaning that the matrix $(x_{i,j})_{i\in I,j\in I}$ is \df{symmetric}. 

We say that an assignment $x$ is a \df{matching} if (1) $x_{i,j}\in \{0,1\}$ for all $(i,j)\in I\times I$, (2) $x$ is symmetric, (3) $ x $ satisfies the unit demand constraints with equality ($\sum_j x_{i,j}= 1$) and (4) $ x $ satisfies the allocation constraints with equality ($\sum_i x_{i,j}= 1$). Define $\C$ to be the convex hull of all matchings.

Note that $\C$ is not equal to the set of symmetric assignments that satisfy the unit demand and allocation constraints, dropping the integrality constraints $x_{i,j}\in\{0,1\}$. \cite{katz1970extreme} proves that the latter set is the convex hull of all matrices of the form $(1/2)(P + P')$ ($P'$ is the transpose of $P$) where $P$ is a permutation matrix with no even cycles greater than 2. A celebrated result of \cite{edmonds1965maximum} provides a characterization of $ \C $, which we use in the proof of Proposition~\ref{prop:Wroomates} below.
 
To operationalize our approach, we need to  work out the set of inequalities $\W$ for the roommates problem. To this end, let $ \mathcal{F} $ be the set of subsets $F\subseteq I\times I$ such that (1) for all $i$, $ (i,i)\notin F $ and (2) for every $ (i,j)\in F $, $ (j,i)\notin F $. For each $ F \in \mathcal{F}$, let $G_F$ be the graph with vertex set $I$ and edge set $\{\{i,j\}:(i,j)\in F \text{ or } (j,i)\in F\}$. Denote the cardinality of the maximum independent edge set of $G_F$ by $k_F$.  For every $ i\in I $, let $ \mathcal{J}_i $ be the set of subsets $ J\subset (\{i\}\times I) \cup (I\times \{i\}) $ such that $ (i,i)\in J $ and for every $ j\neq i $, either $ (i,j)\in J $ or $ (j,i)\in J $ but not both. Then $ \lcs(\C) $ is characterized by the following inequalities.

 \begin{proposition}\label{prop:Wroomates}
 	\[ \lcs(\C) =\bigg(\bigcap_{\os\neq F\in \mathcal{F}} \{x\in\Re^{I\times I}_+ : \sum_{(i,j)\in F}x_{i,j} \leq k_F\}\bigg)\bigcap \bigg(\bigcap_{i\in I, J\in \mathcal{J}_i}\{x\in\Re^{I\times I}_+ : \sum_{(i',j')\in J}x_{i',j'} \leq 1\}\bigg).\]
 \end{proposition}

 \begin{proof}
 	Let $ D $ denote the set on the right-hand side of the proposition. We first prove that $ D\subset \lcs(\C) $.  	
 	For every $x\in D$, consider the matrix $ x' $ obtained by letting $x'_{i,j} = \max\{x_{i,j},x_{j,i}\}$ for all $ (i,j)\in I\times I $. Then $x'$ is symmetric and $x\leq x'$. We prove that $ x'\in D $. For any $ \os\neq F\in \mathcal{F} $, suppose to the contrary that $ \sum_{(i,j)\in F}x'_{i,j}>k_F $. Then we define $ F'\subset I\times I $ such that for every $ (i,j)\in F $, if $ x_{i,j}\ge x_{j,i} $, let $ (i,j)\in F' $, and otherwise let $ (j,i)\in F' $. So $ G_{F'} $ and $ G_F $ have the same (undirected) edge set, and thus $ k_F=k_{F'} $. However, $ \sum_{(i,j)\in F'}x_{i,j}=\sum_{(i,j)\in F}x'_{i,j}>k_{F'} $, which contradicts that $ x\in D $. Similarly we can prove that for every $ i $ and every $ J\in \mathcal{J}_i $, $ \sum_{(i',j')\in J}x'_{i',j'} \leq 1 $. Thus, $ x'\in D $.
 	
 	Now define another matrix $ y $ by (1) for every $ (i,j)\in I\times I $ with $ i\neq j $, set $ y_{i,j}=x'_{i,j} $, and (2) for every $ i\in I $, $ y_{i,i}=1-\sum_{j\neq i}x'_{i,j} $. It is clear that $ y $ is symmetric and that $ x'\le y $ (as $\{i\}\times I\in \mathcal{J}_i$). For any $F\in \mathcal{F},$ $(i,i)\notin F$; hence $\sum_{(i,j)\in F}y_{i,j} = \sum_{(i,j)\in F}x'_{i,j} \leq k_F $. Since $ x' $ is symmetric and $ x'\in D $, for every $ i $ and every $ J\in \mathcal{J}_i $, $ \sum_{(i',j')\in J}y_{i',j'} = 1 $. So $ y\in D $ and it is a bistochastic matrix.
 	
 	Now we prove that $ y\in \C $. \cite{edmonds1965maximum} proves that a symmetric bistochastic matrix $ z $ belongs to $ \C $ if and only if for every $ r\in \mathbb{N} $ and every $ I'\subset I $ with $ |I'|=2r+1 $, $ \sum_{(i,j)\in F}z_{i,j} \leq r $, where $ F\subset I'\times I' $ is such that there does not exist $ (i,i)\in F $ and for every $ (i,j)\in I'\times I'  $ with $ i\neq j $, either $ (i,j)\in F $ or $ (j,i)\in F $ but not both. For any such $F$, $k_F=r$ because $I'$ is odd and we can form $r$ pairs among the $2r$ elements of $I'$ that can be paired. Since $F\in\mathcal{F}$, then,  $ y $ satisfies Edmonds' inequalities and thus $ y\in \C $. Since $ x\le x'\le y $, $ x\in \lcs(\C) $. This means that $ D\subset \lcs(\C) $.
 	
  	To prove $  \lcs(\C)\subset D$, consider any $x\in\C$. Then $x$ is the convex combination of deterministic matchings $x^k$. For each $ \os\neq F\in \mathcal{F} $ and each $i$, there is at most one $j$ with $x^k_{i,j}=1$. By the definition of independent edge set, then  $\sum_{(i,j)\in F}x^k_{i,j}\leq k_F$. So $\sum_{(i,j)\in F}x_{i,j}\leq k_F$. It is clear that $ x $ satisfies the other inequalities related to every $ \mathcal{J}_i $. So $ x\in D $. Then it means that $ \lcs(\C)\subset D $.
 \end{proof}

A pseudo-market equilibrium implies a random matching $x^*$ (a probability distribution over matchings) that is Pareto efficient. 
Of course, $x^*$ needs not be stable in the game theoretic sense, but it corresponds to individual agents' optimizing behavior, as long as these agents take prices as given. Price taking behavior is a plausible assumption in a large centrally-run market for partnerships, like for example a market for roommates in college dormitories. A pseudo-market could be set up by the college, and equilibrium prices could be enforced.

We finalize with a numerical example where stable matchings fail to exist, but where our results deliver an efficient equilibrium (a market stability of sorts). 

\begin{example}[A market for roommates]
  Let $I=\{1,2,3\}$ and $L=3$. For agent $i$, consuming object $l$ is the same as having agent $l$ as her roomate. Suppose that the agents' utilities are 
  \begin{center}
  	\begin{tabular}{c|ccc}
		    & $ 1 $ & $ 2 $ & $ 3 $ \\\hline
		$1 $& $ 0 $ & $ 1 $ & $ 2 $\\
		$2 $& $ 2 $ & $ 0 $ & $ 1 $\\
		$3 $& $ 1 $ & $ 2 $ & $ 0 $
	\end{tabular}
  \end{center}

With these preferences, there are no stable matchings. However, there is a HZ equilibrium. In the equilibrium, the price of the following constraint is two:
\[
x_{2,1}+x_{1,3}+x_{3,2}\le 1,\\
\]the price of the following constraint is one:
\[
x_{1,2}+x_{2,3}+x_{3,1}\le 1,
\]and the price of every other constraint is zero. Then, agent 1's personalized price vector is $(0,1,2)$, 2's personalized price vector is $(2,0,1)$, and 3's personalized price vector is $(1,2,0)$. All of them choose the consumption $(1/3,1/3,1/3)$, and this is the symmetric equilibrium assignment. 
\end{example}

\subsection{Coalition formation}\label{sec:coalitions} The application to roommates can be adapted to a general coalition-formation problem.  Given a set of agents $I$, let $ O $ be the set of all \df{coalitions} from $I$; that is, $ O=2^I\backslash \{\emptyset\} $. A deterministic assignment is a partition of agents into coalitions, and can be represented by a matrix $ x\in \{0,1\}^{NL} $ such that $ x_{i,l}=1 $ if and only if $ i $ joints the coalition $ l\in O $. Unit demand constraints will imply that agents are members of a single coalition. We may then let  $ \C $ be the convex hull of the set of deterministic assignments. Then there exists a pseudo-market equilibrium, and the equilibrium assignment is a probability distribution over coalitions.

\subsection{Combinatorial allocation}\label{sec:combinatorialalloc} Our methods can be used to solve general combinatorial assignment and matching problems \citep{budish2011combinatorial,budish2012matching}. Here we discuss allocation problems in which agents demand a bundle of objects, as in course allocation. In contrast with \cite{budish2011combinatorial}, our emphasis is on random allocations, so there are no problems arising from the lack of convexity of deterministic allocations.

There are obvious supply constraints, stemming from course capacities, but course allocation may exhibit additional, and more problematic, constraints. For example, if a school regards two courses $ l $ and $ l' $ as complements, students must take both of them or neither. Then we have the constraint $ x_{i,l}= x_{i,l'} $. If the school regards two courses $ l $ and $ l' $ as substitutes, so that students have to take at most one of them, then we have the constraint $ x_{i,l}\cdot x_{i,l'}=0 $.

The set of feasible (random) assignments in course allocation problems cannot be easily characterized. In particular, an assignment that seems ex-ante feasible may not be actually implementable. The bihierarchy condition is not met. For example, suppose there are three agents $ 1,2,3 $ and three objects $ a,b,c $. Each object has one copy. The set of bundles is $ O=\{ab,ac,bc\} $. The following random assignment looks ex-ante feasible because it satisfies unit demand constraints of agents and allocation constraints of objects. But it is not feasible because bundles are not independent objects. When a bundle is assigned, the other two bundles become unavailable.
\begin{center}
	\begin{tabular}{c|ccc}
		$i$	  & $ ab $ & $ ac $ & $ bc $  \\ \hline
		$ 1 $   & 1/2                & 0  & 0  \\
		$ 2 $   &   0     & 1/2      & 0 \\
		$ 3 $   & 0 & 0 & 1/2 \\
	\end{tabular}
\end{center}

We can still apply our results by starting from a collection of deterministic assigments. Let $ A $ be the basic set of ``items,'' each of which has a number of copies. Let $ O\subset 2^A $ be the set of bundles under consideration. A deterministic assignment is represented by a matrix $ x\in \{0,1\}^{NL} $ such that $ x_{i,l}=1 $ if and only if $ i $ obtains the bundle $ l\in O $. Let $ \C $ be the convex hull of the set of deterministic assignments. Starting from $\C$, one needs to pre-process $\lcs(\C)$ and our theorem will deliver a pseudo-market equilibrium with the desirable normative properties. 

\section{A market for ``bads''}\label{sec:bads} So far we have assumed that objects are ``goods,'' in the sense that agents' utility functions are strictly increasing. In some applications, however, objects represent duties, or tasks, that agents dislike. Another application is to waste disposal, or pollution. A certain minimum amount of such ``bads'' have to be allocated; the question is to whom, and in which quantities? 

The presence of bads gives rise to floor constraints, but we cannot use our previous methods directly as all agents will choose zero consumption from their consumption space.  We can, however, borrow an idea from the standard model of labor markets: labor supply is often described as consumption of leisure. We endow every agent with a copy of every ``bad,'' and allow them to buy the options of not consuming a bad. Such options become ``goods,'' and our previous methods apply.

Specifically, for every $ l\in O $, $ q_l $ denotes the minimum number of copies of $ l $ that have to be assigned. Every agent can be assigned at most one object (unit demand). If $ \sum_{l\in O}q_l = N  $, then every agent must obtain an object so that the problem becomes the one studied by \cite{HZ1979}.  Assume then that $ \sum_{l\in O}q_l < N  $. For every $ x\in \Delta_{-} $ and every $ i\in I $, $ u_i(x) $ is strictly decreasing in $ x $: if $ x'>x $, then $ u_i(x')<u_i(x) $.

We consider a dual problem $ (I,\tilde{O},\tilde{\Delta}_{-},(\tilde{u}_i)_{i\in I}, (q_{\tilde{l}})_{\tilde{l}\in \tilde{O}}) $ in which
\begin{itemize}
	\item the set of objects is $ \tilde{O}=\{\tilde{l}\}_{l\in O} $ where every $ \tilde{l} $ is an artificial object dual to $ l\in O $, and its supply is $ q_{\tilde{l}}=N-q_l $. When an agent $ i $ consumes an amount $ a $ of $ \tilde{l} $, it is understood that $ i $ consumes $ 1-a $ of $ l $. Because at least $ q_l $ of $ l $ need to be assigned, the number of copies of $ \tilde{l} $ is $ N-q_l $.
	
	\item The consumption space for every agent is $ \tilde{\Delta}_{-}=\{x\in \mathbf{R}^L_+:x_{\tilde{l}}\in[0,1] \text{ for every }l\in \tilde{O}, \sum_{\tilde{l}\in \tilde{O}}x_{\tilde{l}}\in [L-1,L]\} $. So the amount of objects in $ O $ that $ i $ will consume is $ L- \sum_{\tilde{l}\in \tilde{O}}x_{\tilde{l}}\in [0,1]$.
	
	\item Every agent $ i $ has the utility function $ \tilde{u}_i $ such that for every $ x\in \tilde{\Delta}_{-} $, $ \tilde{u}_i(x)=u_i(\mathbf{1}-x) $. When $ u_i $ is (semi-strictly) quasi-concave and strictly decreasing, $ \tilde{u}_i $ is (semi-strictly) quasi-concave and strictly increasing.
\end{itemize}

In the dual problem, agents can consume multiple artificial objects. We impose floor constraints on individual consumption, and can derive the inequalities to characterize $ \lcs(\C)$ as in Section \ref{sec:floorconstraints}. Then Theorem~\ref{thm:existence:noend} applies to give a desirable outcome. We omit the details.

\section{Endowment and $ \al $-slack equilibrium}\label{sec:endowment}

We turn to a version of our model in which objects are initially owned by agents as endowments. Endowments are important in market design when the purpose is to re-assign resources. Often, one wants to improve on an existing allocation. It is then important to be able to respect agents' property rights.\footnote{Re-assignment problems give rise to political economy issues. The most basic issue is to ensure that agents are not hurt in the re-allocation; that their property rights are respected.} Moreover, there are models (such as time banks, briefly discussed in~\ref{sec:timebanks}), in which the agents themselves provide the goods that are to be allocated. 

\subsection{The economy and equilibrium}
 Now an \df{economy} is a tuple $\Gamma=(I,O,(Z_i,u_i,\w_i)_{i\in I})$, where 
\begin{itemize}
	\item $I$ is a finite set of \df{agents};
	\item $O$ is a finite set of \df{objects};
	\item $Z_i\subseteq \Re^\nobjects_+$ is $i$'s \df{consumption space};
	\item $u_i:Z_i\rightarrow \Re$ is $i$'s \df{utility function};
	\item $\w_i\in Z_i$ is $i$'s \df{endowment}.
\end{itemize}

The \df{aggregate endowment} is denoted by $\bar \w = \sum_{i\in I} \w_i$. For every $ l\in O $, $ \bar \w_l $ is the amount of $ l $ in the economy. 

A \df{constrained allocation problem with endowments} is a pair $(\Gamma,\C)$ in which $\Gamma$ is an economy and $\C$ is a set of feasible assignments such that
\begin{enumerate}
	\item $\C$ is a polytope; 
	\item $\w=(\w_i)_{i\in I}\in \C$; that is, $ \w $ is feasible.
\end{enumerate}

A feasible assignment $x\in \C$ is \df{acceptable} to agent $i$ if $u_i(x_i)\ge u_i(\w_i) $; $x$ is \df{individually rational} (IR) if it is acceptable to all agents. We also define a notion of approximate individual rationality: for any $\ep>0 $, $x$ is \df{$\ep$-individually rational} ($ \ep$-IR) if $ u_i(x_i)\ge u_i(\w_i)-\ep$ for all $ i\in I $. 

Let $ \X_i $ and $ \W^* $ be defined as before. We say two agents $i$ and $j$ are of \df{equal type} if $ \w_i=\w_j $, $\X_i=\X_j$, and for all $(a,b)\in \W^*$, $ a_i=a_j $.

In a textbook exchange economy, Walrasian equilibrium assumes that agents' incomes equal the value of their endowments at equilibrium prices. However, when consumption space is bounded, agents may have satiated preferences. Then Walrasian equilibrium may not exist; see Example~\ref{ex:HZexample} in Section~\ref{sec:unitdemandWalrasian}.
 
Our method to solve the nonexistence problem is to introduce an arbitrarily small exogenous budget. Given any price vector $ p $, let $ p_i $ be the personalized price vector faced by $i$, as defined in Section \ref{sec:eqmdefn}. Then for any $ \al\in [0,1] $, we let $ i $'s budget be
\[
\al+(1-\al)p_i\cdot \w_i.
\]
So $i$'s income is a convex combination of the exogenous budget of $1$ used by HZ (and in our model of Section~\ref{sec:model}), and the market value of $ i $'s endowment. 

For any $ \al\in [0,1] $, we say $ (x^*,p^*) $ is an \textbf{$ \al $-slack equilibrium} if
\begin{enumerate}
	\item $ x^*_i\in \arg\max_{x_i\in \mathcal{X}_i}\{u_i(x_i):p^*_i\cdot x_i\leq \al+(1-\al)p^*_i\cdot \w_i\}$;
	\item  $x^*\in \C$;
	\item For any $ c=(a,b)\in \W^* $, $ a\cdot x^*<b$ implies that $ p^*_c=0 $.
\end{enumerate}

\begin{remark} Textbook Walrasian equilibria are $0$-slack equilibria. The pseudo-market equilibria we have already discussed in detail are $1$-slack equilibria. 
\end{remark}

\subsection{Results}

We assume that for each $ c \in \W^* $, $ \sum_{(i,l)\in supp(c)}\w_{i,l}>0 $. A sufficient condition for this assumption is that every agent owns a positive amount of every object. Our next result is a generalization of Theorem \ref{thm:existence:noend}.

\begin{theorem}\label{thm:existence}
	Suppose that agents' utility functions are continuous, quasi-concave and strictly increasing. For any $\al\in (0,1]$:
	\begin{itemize}
		\item There exists an $ \al $-slack equilibrium $(x^*,p^*)$, and $ x^* $ is weakly $\C$-constrained Pareto efficient.
		
		\item If agents' utility functions are semi-strictly quasi-concave, there exists an $ \al $-slack equilibrium assignment $x^*$ that is $\C$-constrained Pareto efficient.
		
		\item Every $ \al $-slack equilibrium assignment is equal-type envy-free.
	\end{itemize}
\end{theorem}

Theorem~\ref{thm:existence} ensures that we can choose $\al\in (0,1]$ arbitrarily, but since prices are endogenous it is not clear that the nominal magnitude of $\al$ has any actual meaning. Our next result shows that it does. In fact, by choosing $\al$ arbitrarily small we ensure that agents' budgets approximate the market values of their endowments. In consequence, the $\al$-slack equilibrium obtained is approximately individually rational.

\begin{theorem}\label{thm:epIR}
	Suppose that agents' utility functions are continuous, semi-strictly
	quasi-concave and strictly increasing. For any $\ep>0$, there is $\al\in(0,1]$ and an
	$\al$-slack equilibrium $(x^*,p^*)$ such that $x^*$ is $\C$-constrained Pareto efficient and 
	\[
	\max\{u_i(y):y\in \X_i\text{ and } p^*_i\cdot y\leq
	p^*_i\cdot \w_i \} - u_i(x^*_i)  <\ep.
	\] In particular, $x^*$ is $\ep$-individually rational.
\end{theorem}

\subsection{The Hylland and Zeckhauser example}\label{sec:unitdemandWalrasian}

A major application of Theorem~\ref{thm:existence} is to the object allocation model under the supply and unit demand constraints. That is, $\X_i = \Delta_{-}$ for every $ i $ and $x\in \C$ if and only if $\sum_i x_{i} = \bar \w$. There are exactly $\nobjects$ inequalities in $\W^*$, one for each object $l$, expressing that $\sum_i x_{i,l}\leq \bar\w_l$. All agents face equal personalized prices, so we write $ p_l $ for the price of $ l $.

We present an example due to
\cite{HZ1979} showing that a Walrasian equilibrium (a $0$-slack equilibrium) may not exist, and show how the symmetric Pareto efficient assignment
in the example can be sustained as an $\al$-slack
equilibrium with any $\al\in (0,1]$.

\begin{example}\label{ex:HZexample}
	Given is an economy with three agents $ 1,2,3 $ and two objects $a,b$. Object $ a $ has one copy and $b$ has two copies. Agents have the following von-Neumann Morgenstern utilities:
	\begin{center}
		\begin{tabular}{c|cc}
			$i$	& $ u_{i,a} $ & $ u_{i,b} $  \\ \hline
			$ 1 $   & 100              & 1            \\
			$ 2 $   & 100              & 1              \\
			$ 3 $   & 1             & 100           \\
		\end{tabular}
	\end{center}
	Endowments are $\w_i=(1/3,2/3)$ for $i=1,2,3$.
\end{example}

\begin{claim}
	There is no Walrasian equilibrium in Example \ref{ex:HZexample}.
\end{claim}
\begin{proof}
	Suppose (towards a contradiction) that $(x,p)$ is a Walrasian
	equilibrium. Suppose first that $p_b>0$ and normalize it to one. Then all agents have the same positive budget. If $ p_a=0 $, then 1 and 2 would each buy one copy of $ a $, which is a contradiction. So $ p_a $ must be positive. The preferences of agents imply that 1 and 2 must each obtain a half of $ a $. Therefore, $ 1/3p_a+2/3\ge 1/2p_b $, and we obtain $ p_a\le 4 $. However, if $ p_a<4 $, $ 1 $ and $ 2 $ would spend all of their budgets on $ A $, and each obtain more than a half of $ a $, which is a contradiction. So it must be that $ 1/3p_a+2/3= 1/2p_a $ and $ p_a= 4 $. But this means that at most 3 demands $ b $ and $ b $ must have excess supply, which is a contradiction.
	
	Now suppose $p_b=0$ and $ p_a>0 $. Then 3 must obtain one copy of $b$. Since $ p_a $ is positive, 1 and 2 must each obtain a half of $ a $. However, their budget $ 1/3p_a $ cannot afford such a consumption.   \end{proof}

Consider the assignment $x$ defined by:
\begin{center}
	\begin{tabular}{c|cc}
		$i$	  & $ x_{i,a} $ & $ x_{i,b} $  \\ \hline
		$ 1 $   & 1/2                & 1/2     \\
		$ 2 $   & 1/2             & 1/2  \\
		$ 3 $   & 0 & 1 \\
\end{tabular}\end{center}

\begin{claim}\label{prop:HZexalWE}
	For any $\al\in(0,1]$, the price vector $ p = (\frac{6\al}{1+2\al},0) $ and the assignment $ x $ constitute an
	$\al$-slack equilibrium in Example \ref{ex:HZexample}.
\end{claim}

\begin{proof} For any $\al\in (0,1]$ and $i=1,2,3$,
	 \[ 
	\al + (1-\al) p\cdot \w_i = 
	\al + (1-\al)\frac{2\al}{1+2\al}
	=\frac{3\al}{1+2\al} = p\cdot x_i.
	\]
	With such budgets, agents 1 and 2 can only afford a $1/2$ share of $ a $ and a
	$1/2$ share of $ b $, which is the best consumption for them. Agent 3 chooses a copy of $ b $ for free.
\end{proof}

Note that in the above $\al$-slack equilibrium, the endogenous value of agents' endowments is
$2\al/(1+2\al)$. So the value of the exogenous part
of the budget relative to the value of the 
endogenous part is \[ 
\dfrac{\al}{(1-\al)\dfrac{2\al}{1+2\al}}\rightarrow \frac{1}{2}
\] as $\al\rightarrow 0$. So when $\al$ shrinks to zero, the value of
the exogenous income is not negligible. 
In the same spirit, the following proposition shows that the average
endogenous budget will always be below the exogenous budget of one. This means that the economy needs outside ``money.''

\begin{proposition}\label{prop:avgincom}
	If $(x,p)$ is an $\al$-slack equilibrium,
	then \[ 
	\frac{1}{N} \sum_{i=1}^N p\cdot \w_i \leq 1
	\] \end{proposition}

\begin{proof}
	Note that $p\cdot (x_i-\w_i)\leq \al(1-p\cdot \w_i)$.  Sum over $i$ to
	obtain: \[ 
	0=  p\cdot \left(\sum_i x_i-\bar\w\right) \leq  \al (N - p\cdot \bar \w).
	\]
\end{proof}

\subsection{A market-based fairness property}

In the object allocation model of Section \ref{sec:unitdemandWalrasian}, agents face identical prices. It is possible to use our result to develop a kind of fairness property in the presence of endowments. Fairness, in the sense of absence of envy, is generally incompatible with individual rationality. Imagine an economy with two objects, where both agents prefer object 1 over object 2, and all the endowment of object 1 belongs to agent 1. Then, in any allocation, there will either be envy, or agent 1's individual rationality will be violated. So fairness has to be amended to account for the presence of endowment.\footnote{The paper by \cite{echenique2020fairness} deals exclusively with this problem, but proposes a very different solution.}

In the object allocation model with supply and unit demand constraints, in any $\al$-slack equilibrium, if agent $i$ envies agent $j$ then it must be that $j$'s endowment is worth more than $i$'s at equilibrium prices. In a sense, this means that agents collectively value $j$'s endowment more than $i$'s. Our next result formalizes this idea. 

\begin{proposition}
  \label{prop:envyslackwalras}
In the object allocation model with supply and unit demand constraints, suppose that agents' utility functions are concave and
$C^1$.\footnote{A function with domain $D$ is  $C^1$ if it can be extended to a continuously differentiable function defined on an open set that contains $D$.} Let $(x,p)$ be an $\al$-slack equilibrium. Denote by
$S = \{i:u_i(x_i)=\max \{u_i(z_i):z_i\in\Delta_{-}\} \}$ the set of
\df{satiated} agents, and by $U=I\setminus S$ the set of others.  Suppose that
$\sum _{i\in U}  x_i\gg 0$. 

If $i$ envies $j$ in $x$ ($u_i(x_j)>u_i(x_i)$), then
$p\cdot\w_j>p\cdot\w_i$, and there exist welfare weights $\ta\in \Re^U_{++}$ such that if 
\[
v(t) = \sup \{
\sum_{i\in U} \ta_i u_i(\tilde x_i)  : (\tilde x_i)\in\Delta_{-}^U \text{ and }
\sum_{i\in U} \tilde x_i \leq \bar\w  + t (\w_i-\w_j) - \sum_{i\in S} x_i  
\},\] 
then $(x_i)_{i\in U}$ solves the problem for $v(0)$, and 
$v(t)<v(0)$ for all $t$ small enough.
\end{proposition}

The meaning of Proposition~\ref{prop:envyslackwalras} is that if an agent $i$ envies agent $j$ then $j$'s endowment is more valuable than $i$'s in two senses. First, it is more valuable at equilibrium prices. Second, the higher price valuation translates into a statement about how much agents value the endowments. In particular, $j$'s endowment is more valuable than $i$'s to a coalition of players $U$ (a coalition that includes $i$!). It is more valuable to $U$ in the sense that there are welfare weights for the members of $U$ such that a change in agents' endowment towards having more of $i$'s endowment and less of $j$'s leads to a worse weighted utilitarian outcome. The result requires that $\sum _{i\in U}  x_i\gg 0$  simply to ensure that when we subtract $\w_j$ we do not force some agent to consume negative quantities of some object.\footnote{The $\sum _{i\in U}  x_i\gg 0$ hypothesis in Proposition~\ref{prop:envyslackwalras} is stronger than what we
need. It suffices that if $\w_{j,l}>0$ then $\sum _{i\in U}  x_{i,l} > 0$.}

\subsection{A market for time exchange}\label{sec:timebanks}
In organizations such as \textit{time banks}, members exchange time and skills without using monetary transfers.\footnote{See \cite{Andersson2018} for more description of real-life time banks.} A time exchange problem can be described as an object allocation model with endowments. Formally, $ O $ is the set of service types. For each agent $ i $ and service $ l $, $ \w_{i,l} $ is the amount of $ l $ that $ i $ can provide. We could require that $ \sum_{l\in O}\w_{i,l}\le 1 $ for all $ i $ and every agent's demand be no greater than one. Here ``one'' could mean one day, one week, or one month. Services can be regarded as divisible because time is divisible. But of course, in real life time is often measured in integers such as hours, days, or weeks. Theorem \ref{thm:existence} implies that we can find a market equilibrium to the problem. The value of an agent's endowment at equilibrium prices shows how much his endowment is valued by all agents. When the value is higher, the agent is rewarded by a better assignment.

\section{Related Literature}\label{sec:relatedliterature}
Constrained resource allocation has received a lot of attention in recent years. \cite{budish2013designing} identify the bihierarchy structure of constraint blocks in the assignment matrix as the sufficient and necessary condition for implementation. \cite{akbarpourapproximate} extend this result by relaxing some constraints and considering approximate implementation. We circumvent the implementation issue by taking the set of implementable assignments as the primitive. Budish et al. allow for floor constraints in implementation but rule out them in their applications. In their extension of the pseudo-market mechanism, they consider column constraints, row constraints and sub-row constraints. By incorporating all row and sub-row constraints into agents' consumption spaces, they prove the existence of equilibria much like Hylland and Zeckhauser's. Their extension is a special case of ours. We can deal with more general constrains on both rows and columns, and allow for floor constraints. When there are no floor constraints, we directly price ceiling constraints, and when there are floor constraints, we translate floor constraints into a different set of ceiling constraints.

\cite*{ehlers2011school} focus on the problem of controlled school choice (which was introduced by \cite{abdulkasonmez}), whereby children have to be assigned seats at different schools to satisfy some diversity objective.\footnote{Controlled school choice is also investigated by, among others, \cite{ehlers2010}, \cite{hafalir2013effective}, \cite{kominers2013designing}, \cite{westkamp2013analysis}, \cite{echenique2015control}, \cite{fragiadakis2017improving}, \cite{aygun2017college}, and \cite{nguyen2017stable}.}  \cite*{kamada2015efficient} are mainly (but not exclusively) motivated by the problem of allocating doctors to hospitals to satisfy geographic quotas. The objective of the quotas is to avoid an excessive concentration of doctors in urban areas.\footnote{See \cite{kamada2017recent} for an overview.} Both papers proceed by adapting the notion of stability to capture the presence of constraints, and to add structure to the constraints being considered. To address more general constraints, \cite*{kamada2019fair} relax stability and focus on feasible, individually rational, and fair
assignments. They demonstrate that the class of general upper-bound constraints on individual schools are the most permissive constraints under which a student-optimal fair matching exists. That class rules out floor constraints. Our paper can deal with the same kinds of constraints in the above papers, but we follow a different methodological tradition. Instead of a two-sided game-theoretic matching model, we consider object allocation and propose a competitive equilibrium solution. The above  papers also investigate the role of incentives in their mechanisms. We expect our pseudo-market mechanism to be incentive compatible in large markets, but we choose to focus on existence, efficiency and fairness.\footnote{\cite{he2018pseudo} prove the asymptotic strategy-proofness of their pseudo-market mechanism.} 

The recent work of \cite{balbuzanov2019constrained} considers a version of the probabilistic serial mechanism for object allocation subject to constraints. Like us, he works on a one-sided object allocation model, but the focus on probabilistic serial makes his analysis clearly distinct from ours. We borrow from this paper the idea, expressed in Lemma~\ref{lem:nonnegativelinearineq}, allowing us to focus on non-negative linear inequalities.

The use of  markets over lottery shares to solve centralized allocation problems was first proposed by \cite{HZ1979}. They assume no constraints other than unit demands and limited supply. They impose a fixed income for each agent, independent of prices.  They also emphasize that equilibrium may not be efficient, and introduce the ``cheapest bundle'' property that we employ as well in our version of the first welfare theorem.  Many other papers have followed Hylland and Zeckhauser in analyzing competitive equilibria as solutions in market design; see for instance, \cite{budish2011combinatorial}, \cite{ashlagi2015optimal}, \cite{hafalir2015welfare}, \cite{he2018pseudo}. \cite{mirallesfoundations} establish the second welfare theorem for the market with satiated preferences and token money: every Pareto efficient assignment may be supported in a Walrasian equilibrium with properly chosen budgets. But none of these papers consider constrained allocation problems.\footnote{\cite{he2018pseudo} considers priority-based  constraints, which are  different from the class of constraints studied here. \cite{mirallesfoundations} does not focus on constraints, but can accommodate linear and individual-agent constraints.}

Hylland and Zeckhauser make the point that an equilibrium may not exist in a model with endowments. Like us, 
\cite{mas1992equilibrium}, \cite{le2017} and \cite{mclennan2018}
also propose to avoid the non-existence issue by means of a hybrid income between the exogenous budget and the endogenous Walrasian income.  A version of the hybrid model was first introduced by \cite{mas1992equilibrium}, who presents an existence result with income that is the sum of a fixed income and a price-dependent income. His result requires the first component to be determined endogenously as part of the fixed point argument in the equilibrium existence result. Aside from the presence of constraints, our result differs from his by allowing us to obtain approximate individual rationality with a small exogenous $\al$.
In \citeapos{le2017} notion of equilibrium, two identical objects may have different prices. As a consequence, there may be envy among identical agents, and it may be necessary for some agents to purchase a more expensive copy of an object when a cheaper one is available.\footnote{In Example \ref{ex:HZexample}, a Le's equilibrium is as follows. Let $p=(100,1,\frac{101}{2})$ be a price vector in which the latter two elements are the prices of the two copies of B. Then all agents
	have an income of
	$101/2$. The unique optimal bundle for agents 1 and 2 is $x_i=(1/2,1/2,0)$. Agent $3$
	is willing to spend all his income on buying the more expensive copy of $B$, so $x_3=(0,0,1)$.
	
	Consider a variation of the example in which endowments become $\w_1=(1/3,1/2,1/6)$,
	$\w_2=(1/3,1/6,1/2)$, and $\w_3=(1/3,1/3,1/3)$. Then  $p=(100,1,\frac{101}{2})$ is still an equilibrium price, with 
	$x_1= (\frac{5}{12},\frac{7}{12},0)$,
	$x_2= (\frac{7}{12},\frac{5}{12},0)$, and
	$x_3=(0,0,1)$ being the equilibrium assignment. Observe that agent 1 envies 2, despite they have the same
	utility and the same endowment: $1/3$ of $ A $ and $2/3$ of $ B $.}
Envy among equals is problematic for normative reasons, and it is hard to implement such equilibria in a decentralized fashion.\footnote{One could   interpret different prices for different copies of the same object as   a novel endogenous transfer scheme, but we are unaware of a normative defense of this idea.} 

\cite{mclennan2018} presents an existence result for equilibrium with ``slack'' in a general model that allows for production and encompasses our model as a special case. But his notion of equilibrium with slack differs from ours in important ways. Agents in his (and our) model may be satiated, and his notion of slack controls the distribution of transfers from satiated agents who spend less than their income to unsatiated agents. In contrast, our $\al$ parameter controls the role of endowments, allowing for $\al$ to specify the weight of equal incomes vs.\ (unequal) endowments. In fact, it is possible to construct an example to illustrate the difference between the two notions of equilibrium. In the example no agents are satiated, so the slack in McLennan's notion has no role to play, and his equilibrium allocations are independent of $\al$; in contrast, our equilibrium allocations range from equal division to the autartical consumption of endowments, as $\al$ ranges from placing all weight on the exogenous income, to placing all weight on initial endowments.\footnote{We are grateful to Andy McLennan for this example, which can be found in his paper.}

\cite{kojima2018job} and \cite{gulpesendorferzhang} consider market equilibrium in economies with gross substitutes utilities and constraints. Kojima et.\ al characterize the constraints that preserve the gross substitutes property of firms' demands in a transferable utility model (\citeapos{kelso1982job} job matching model). Gross substitutes ensure equilibrium existence, and the authors show that the constraint structures have to take the form of ``interval constraints.''  Gul et.\ al prove the existence of equilibrium in economies with a finite number of indivisible objects, and limited transfers or no transfers. They show that with limited transfers or no transfers, equilibrium requires random allocations and can be approached by the equilibrium with full transfers. They also show that
equilibrium allocations satisfying certain constraints can be constructed by building these constraints into utility functions or
incorporating them into a production technology. Different from them, we price constraints and can accommodate more general preferences and constraint structures.

Related to our applications, \cite{manjunath2016fractional} proposes a competitive equilibrium notion for a two-sided fractional matching market. The double-indexed price system in his notion resembles our personalized price system, but he needs to deal with both sides' preferences. As a consequence, his equilibrium exists when there are transfers, but only approximately exists when transfers are forbidden. 
\cite{Andersson2018} propose a time exchange model in which each agent provides a distinct service and has dichotomous preferences. They propose a priority mechanism to maximize the amount of exchanges among agents. Differently, in our model of time exchange an agent can provide multiple services and different agents can provide the same service. Agents can express
richer preferences and the prices in our market solution reveal on which service agents have more demand. \cite{bogomolnaia2017competitive,bogomolnaia2019dividing} study the competitive equilibrium allocation of a mixed manna that contains ``goods'' and ``bads''. They prove that an equilibrium always exists. Our model is different than theirs in that agents have unit-demand constraints. So their existence result does not hold in our paper.

Finally, the recent work by \cite{rootincentives} looks at constrained allocation from a mechanism design perspective. They allow for very general constraint sets, and prove a characterization of group strategy-proof rules.

\section{Proof of Theorem \ref{thm:existence:noend} and Theorem \ref{thm:existence}}\label{sec:proofexistence}

We first prove Theorem \ref{thm:existence} by assuming that all utility functions are semi-strictly quasi-concave. We then explain in Remark \ref{rmk:cheapb} the differences when utility functions are only quasi-concave. After that, in Remark \ref{rmk:proof:Thm1} we explain how the proof works for Theorem \ref{thm:existence:noend}.

With an abuse of notation, we write $\sum_{l\in O}p_{i,l}x_{i,l}$ as
$p_i\cdot x_i$. 
For each $ c\in \W^* $, we write $c=(a^c,b^c)$. We have assumed that $ \sum_{(i,l)\in supp(c)}\w_{i,l}>0 $. It implies that $ \sum_{(i,l)\in supp(c)}a^c_{i,l}\w_{i,l}>0 $. 

We define a price ceiling 
\[
\bar{p}= \frac{2N}{\min_{c\in \W^*} \sum_{(i,l)\in supp(c)}a^c_{i,l}\w_{i,l}},
\] and a price space $ \mathcal{P}=[0, \bar{p}]^{\W^*} $.

Given $ \al\in (0,1] $,  for every $ p\in \mathcal{P} $, define
\begin{align*}
v_i & = \max \{u_i(x_i):x_i\in\X_i\}, \\
B_i(p,\al) & = \{x_i\in\X_i:p_i\cdot x_i\le \al+(1-\al)p_i\cdot \w_i \},  \\
d_i(p) & = \argmax \{u_i(x_i): x_i\in B_i(p,\al)\}, \\
\ul d_i(p)& = \argmin\{ p\cdot x_i: x_i\in d_i(p)\}, \\
V_i(p) & = \max \{u_i(x_i): x_i\in B_i(p,\al)\}.
\end{align*}

\begin{lemma}\label{lem:bcbinding} 
	If $V_i(p)<v_i$ then $d_i(p)=\ul d_i(p)$. 
\end{lemma}
\begin{proof} Let $x_i\in d_i(p)$. We shall prove that $p_i\cdot x_i= \al+(1-\al)p_i\cdot \w_i$, which means we are done because it implies that all bundles in $d_i(p)$ cost the same at prices
	$p$. Let $z_i\in\X_i$ be such that $u_i(z_i)=v_i>u_i(x_i)$. For any $\ep\in (0,1)$, since $ \X_i $ is convex, $\ep z_i +
	(1-\ep)x_i\in\X_i$. By the
	semi-strict quasi-concavity of $u_i$, $u_i(\ep z_i  + (1-\ep)x_i)> u_i(x_i)$. This means that, for any $\ep\in
	(0,1)$,
	\[
	\ep p_i\cdot z_i+(1-\ep) p_i \cdot x_i> \al+(1-\al)p_i \cdot \w_i.
	\] But this is only possible, for arbitrarily small $\ep$,  if
	$p_i\cdot x_i\ge \al+(1-\al)p_i\cdot \w_i$. Since $x_i\in B_i(p,\al)$, we have $p_i\cdot x_i= \al+(1-\al)p_i\cdot \w_i$. 
\end{proof}

\begin{lemma}\label{lem:bliss} If $V_i(p)=v_i$, then 
	\[\ul d_i(p) =
	\arg\min\{p_i\cdot x_i : u_i(x_i)=v_i\text{ and } x_i\in\X_i \}.
	\] 
\end{lemma}

\begin{proof} Let $x_i\in \ul d_i(p)$. Then for any $z_i\in\X_i$
	with $p_i\cdot z_i< p_i \cdot x_i$, we have $z_i\in B_i(p,\al)$. So 
	$u_i(z_i)<v_i$ by definition of $\ul d_i$. Therefore, if  $z_i\in \argmin\{p_i\cdot x_i : u_i(x_i)=v_i\text{ and } x_i\in\X_i \}$, then
	\[
	p_i \cdot z_i= p_i \cdot x_i,
	\] and therefore 
	\[\ul d_i(p) \supseteq
	\arg\min\{p_i\cdot x_i : u_i(x_i)=v_i\text{ and } x_i\in\X_i \}.
	\] 
	The converse set inclusion follows similarly because if $x_i$ is not
	in the right-hand set, there would exist $z_i\in\X_i$
	with $p_i \cdot z_i< p_i \cdot x_i$ and $u_i(z_i)=v_i$, which is not possible as
	such $z_i$ would be in $B_i(p,\al)$. 
\end{proof}

\begin{lemma}\label{lem:diuhc}
$d_i$ is upper hemi-continuous.
\end{lemma}
\begin{proof} 
  Let $(x^n,p^n)\rightarrow (x,p)$, with $x^n\in d_i(p^n)$. Suppose
  that there is $x'\in B_i(p,\al)$ with $u_i(x')> u_i(x)$. If $p_i\cdot x'<
  \al + (1-\al)p_i\cdot \w_i$, then this strict inequality will be true
  for $p^n$ with $n$ large enough; a contradiction, as $u_i$ is
  continuous. If $p_i\cdot x'= \al + (1-\al)p_i\cdot \w_{i}$, then $\al>0$
  implies that $p_i\cdot x'>0$. Then there is $\la\in (0,1)$ large
  enough that $u_i(\la x')> u_i(x)$, $p_i\cdot (\la x')< p_i\cdot x'$, and
  $\la x'\in\X_i$ (recall that the construction of $\X_i$ ensures that this is the case)). The argument for the case of a strict
  inequality then applies. 
\end{proof}

\begin{remark} Lemma~\ref{lem:diuhc} uses crucially that $\al>0$. 
\end{remark}

\begin{lemma}\label{lem:uhc} $\ul d_i(p)$ is upper hemi-continuous. 
\end{lemma}
\begin{proof} 
	 To prove upper hemi-continuity, we shall prove that $\ul d_i$ has a closed graph. Let
	$(x_i^n,p^n)\rightarrow (x_i,p)$ with $x_i^n\in\ul d_i(p^n)$ for all $n$. 
	
	First, consider the case where $V_i(p)<v_i$. By the maximum theorem,
	$V_i$ is continuous, so $V_i(p^n)<v_i$ for all large enough $n$. Then
	Lemma~\ref{lem:bcbinding} implies that $x_i\in\ul d_i(p)$ as $d_i$ is
	upper hemi-continuous. 
	
	Second, consider the case where $V_i(p)=v_i$. We know that $x_i\in
	d_i(p)$ as $d_i$ is upper hemi-continuous. Suppose (towards a
	contradiction) that $x_i\notin \ul d_i(p)$. Then 
	there is $y_i\in d_i(p)$ with 
	\[p_i\cdot y_i< p_i\cdot x_i \leq \al + (1-\al)p_i\cdot \w_i.\]
	Then for all $n$ large
	enough,
	\[
	p^n_i\cdot y_i<\al + (1-\al)p_i \cdot \w_i.
	\]
	
	Since $y_i\in d_i(p)$ and $V_i(p)=v_i$, $u_i(y)=v_i$. This means that
	$V_i(p^n)=v_i$ for all $n$ large enough, as $y_i\in B_i(p^n,\al)$. Then, by Lemma~\ref{lem:bliss}, $x_i^n\in \argmin\{p^n_i\cdot x_i:
	u_i(x_i)=v_i\text{ and } x_i\in\X_i \}$ for all $n$ large
	enough. But the  correspondence  \[p\mapsto
	\argmin\{p_i\cdot x_i: u_i(x_i)=v_i\text{ and } x\in\X_i \}.\] 
	is upper hemicontinous, by the maximum theorem. So 
	\[
	x_i\in \argmin\{p_i\cdot x_i : u_i(x_i)=v_i\text{ and } x\in\X_i \},
	\] which is a contradiction.
\end{proof}

It is easy to see that $ d_i(p) $ is nonempty, compact- and convex-valued. So $ \ul d_i(p) $ is also nonempty, compact- and convex-valued. 
For every $c\in \W^* $, define the aggregate demand on $ c $ by
\[
D_c(p)=\sum_{(i,l)\in supp(c)} a^c_{i,l}\ul d_{i,l}(p) = \cup \{
a^c\cdot x:x\in \times_i \ul d_i(p)\}. 
\] 
Define the aggregate demand correspondence by
\[
D(p)=(D_c(p))_{c\in \W^*},
\] and the excess demand correspondence by
\[
z(p)=D(p)-\{\textbf{b}\},
\]where $ \textbf{b}=(b^c)_{ c\in \W^*} $.

Consider the correspondence $ \phi:\mathcal{P}\rightarrow \mathcal{P} $ defined by
\[
\phi_c(p)=\{\min\{\max\{0,z_c+p_c\},\bar{p}\}:z\in z(p)\} \text{ for all }c\in \W^*.
\]

 $ D(p) $, and therefore $ z(p) $, are upper hemi-continuous, convex-valued, and compact-valued. Thus, $ \phi $ is upper hemi-continuous, convex-valued and compact-valued. By Kakutani's fixed point theorem, there exists $ p^*\in \mathcal{P} $ with $ p^*\in\phi(p^*) $. 

Note that there exists $ z^*\in z(p^*) $ such that
\begin{equation}\label{eq:FP}
	p^*_c=\min\{\max\{0,z^*_c+p^*_c\},\bar{p}\} \text{ for all }c\in \W^*.
\end{equation}

Choose $x^*\in\Re^{\nagents \nobjects}_+$ such that $x^*_i\in \ul d_i(p^*)$ for all
$i$ and $a^c\cdot x^* - b^c = z^*_{c}$ for all $c\in\W^*$. We
shall prove that $ (x^*,p^*)$ is an $\al$-slack Walrasian
equilibrium. 

\begin{lemma}\label{lem:almostWL}
	$p^*\cdot z^*\geq 0$.
\end{lemma}
\begin{proof}
	If  $p^*\cdot z^*<0$, then there is some $ c\in \W^* $ with $p^*_c>0$ and
	$z^*_c<0$. By Equation~\ref{eq:FP}, then, $p^*_c=p^*_c+z^*_c$,
	which is not possible as $z^*_c<0$.
\end{proof}

\begin{lemma}\label{lem:boundary}
	$p^*_c<\bar p$ for all $c\in \W^*$.
\end{lemma}
\begin{proof} Suppose towards a
	contradiction that there exists $c^*\in \W^*$ for which
        $p^*_{c^*}=\bar p$. Then 
	$p^*_c>0$. Now, $ x^*_i\in B_i(p^*,\al) $ means that 
	\[
	p^*_i\cdot x^*_i\le \al+(1-\al)p^*_i\cdot\w_i,
	\]which is equivalent to
	\[
	p^*_{i}\cdot (x^*_{i}-\w_{i})\le \al(1-p^*_i\cdot \w_i).
	\]
	
	Summing over $i$, we obtain that 
	\[
	\sum_{i\in I}p^*_{i}\cdot (x^*_{i}-\w_{i})\le \al(N-\sum_{i\in I}p^*_i\cdot \w_i),
	\]which is equivalent to
	\begin{equation} \sum_{i\in I}\sum_{l\in O}
          p^*_{i,l}(x^*_{i,l}-\w_{i,l})\le \al (N-\sum_{i\in
            I}\sum_{l\in O} p^*_{i,l}\w_{i,l}) \label{eq:stepone}
\end{equation}
	
	Note that 
	\begin{align}
		\sum_{i\in I}\sum_{l\in O} p^*_{i,l}\w_{i,l}&=
                \sum_{i\in I}\sum_{l\in O} \big(\sum_{c\in
                  \W^*:(i,l)\in supp(c)}a^c_{i,l}p^*_c\big)\w_{i,l}
                \notag \\
		&=\sum_{c\in \W^*} p^*_c\big( \sum_{(i,l)\in
                  supp(c)}a^c_{i,l}\w_{i,l}\big) \label{eq:steptwo} 
	\end{align}
	
	Now, by definition of $\bar p$, we have that  
	\[
	\sum_{c\in \W^*} p^*_c\big( \sum_{(i,l)\in supp(c)}a^c_{i,l}\w_{i,l}\big) \geq \bar p  \big( \sum_{(i,l)\in supp(c^*)}a^{c^*}_{i,l}\w_{i,l}\big)
	 \ge \bar p  \min_{c\in \W^*}\big( \sum_{(i,l)\in supp(c)}a^{c}_{i,l}\w_{i,l}\big)
	=2N.
	\]
	
	 Thus, using this inequality and equations~\eqref{eq:stepone}
         and~\eqref{eq:steptwo}, we obtain that 
	 \begin{equation}
	\sum_{i\in I}\sum_{l\in O} p^*_{i,l}(x^*_{i,l}-\w_{i,l})\le \al (N-\sum_{i\in
            I}\sum_{l\in O} p^*_{i,l}\w_{i,l})<0. \label{eq:negative1}
	 \end{equation}

	 On the other hand,
\begin{align}
p^*\cdot z^* 
& =\sum_{c\in \W^*} p^*_c\big(\sum_{(i,l)\in
  supp(c)}a^c_{i,l}x^*_{i,l}-b^c\big).  \notag \\
& \leq \sum_{c\in \W^*} p^*_c\big(\sum_{(i,l)\in
  supp(c)}a^c_{i,l}x^*_{i,l} - \sum_{(i,l)\in  supp(c)}a^c_{i,l}\w_{i,l}\big) \label{eq:blargerthanw} \\
& = \sum_{i\in I}\sum_{l\in O}\sum_{\{c\in\W^*: (i,l)\in supp(c)\}} p^*_c
a^c_{i,l} (x^*_{i,l}-\w_{i,l}) \notag \\
& = \sum_{i\in I}\sum_{l\in O} p^*_{i,l}(x^*_{i,l}-\w_{i,l}) \label{eq:aclarando} \\
& < 0 \label{eq:negative2},
	 \end{align}
where~\eqref{eq:blargerthanw} follows because for each $ c\in \W^* $,
$ \sum_{(i,l)\in supp(c)}a^c_{i,l}\w_{i,l}\le b^c $,
\eqref{eq:aclarando} follows as 
	 \[
	 p^*_{i,l}=\sum_{\{c\in\W^*: (i,l)\in supp(c)\}} p^*_c a^c_{i,l},
	 \] and  \eqref{eq:negative2} follows from \eqref{eq:negative1}. 

Finally, \eqref{eq:negative2} is absurd as it contradicts
Lemma~\ref{lem:almostWL}.   
\end{proof}

\begin{proof}[\color{blue}Proof of Theorem~\ref{thm:existence}\color{black}]
	We claim that $ (x^*,p^*) $ is an $\al$-slack Walrasian equilibrium. If $p^*_c>0 $, since $p^*_c<\bar{p} $, then $ p^*_c=z^*_c+p^*_c$, which implies $ z^*_c=0 $. If $ p^*_c=0 $, then $z^*_c+p^*_c\le 0 $, which implies $ z^*_c\le 0 $. Recall that
        $z^*_c = a^c\cdot x^*-b^c$. So this implies that $ x^* $
        satisfies all inequalities in $ \W^*$. By definition of
        $\X_i$, $x^*$ satisfies then all inequalities in
        $\W$. Hence, \[ x^*\in \lcs(\C).\]
        Moreover, if $ z^*_c<0 $, it must be that $ p^*_c=0 $, as $p^*_c>0 $ implies $ z^*_c=0 $.  
	
It remains to show that $ x^*\in \C$. Suppose to the contrary that $ x^*\notin \C $. Since $x^*\in \lcs(\C)$, there exists $ x'\in\C$ such that $ x^*\le x' $. Then $x^*\neq x'$. So there is $(i^*,l^*)\in I\times O $ with $ x^*_{i^*,l^*}<x'_{i^*,l^*}$. By definition of $\C$, $x'_{i^*}\in Z_i$. 

Consider $y_{i*}$ defined as $y_{i^*,l} = x^*_{i^*,l}$ for all $l\neq l^*$, and  $y_{i^*,l^*} = x'_{i^*,l^*}$. Since $x'\in \C$ and $x'_{i^*,l^*}>0$, $l^*$ cannot be a forbidden object for $i^*$. Hence, $x'_{i^*}\in \X_{i^*}$ and therefore $y_{i^*}\in \X_i$. 

Moreover, for any $c\in \W^*$, if $(i^*,l^*)\in \supp(c)$ then $a^c\cdot x^*<a^c\cdot x'\leq b^c$ and therefore $z^*_c<0$ ($c$ must not be binding at $x^*$). Hence $p^*_c=0$. In consequence, 
\begin{align*}
\sum_{l\in O} p^*_{i^*,l} y_{i^*,l} & = \sum_{l\neq l^*} p^*_{i^*,l} y_{i^*,l} + \big(\sum_{c\in \W^*}\underbrace{p^*_{c} a^c_{i^*,l^*}}_{=0} \big)y_{i^*,l^*}\\
&=\sum_{l\neq l^*} p^*_{i^*,l} x^*_{i^*,l}\\
&\leq \al + (1-\al)p^*_i \cdot \w_i.
\end{align*}
Thus $y_{i^*}\in B_{i^*}(p^*,\al)$ and $x^*_{i^*}< y_{i^*}$, contradicting the strict monotonicity of $u_{i^*}$ and that $ x^*_i\in d_i(p^*) $. 

We next prove that $ x^* $ is $ \C $-constrained Pareto efficient. Suppose towards a contradiction that $ x $ is an feasible assignment that Pareto dominates $ x^* $. Given that $x\in \C$, $x_i\in \X_i$. Then, for all $i\in I $, $ u_i(x_i)\ge u_i(x^*_i) $. So by definition of $\ul d_i$ we have that \[
	p^*_i \cdot x_i\ge p^*_i \cdot x^*_i.
	\] And for some $ j\in I $, $ u_j(x_j)> u_j(x^*_j) $, so by utility maximization, 
	\[ p^*_j \cdot x_j> p^*_j \cdot x^*_j.
	\]
	
	Thus,
	\[
	\sum_{i\in I}p^*_i \cdot x_i > \sum_{i\in I}p^*_i \cdot x^*_i.
	\]
	
	This is equivalent to	
	\[
	\sum_{c\in \W^*}p^*_c\bigg(\sum_{(i,l)\in supp(c)}a^c_{i,l}x_{i,l}\bigg) > \sum_{c\in \W^*}p^*_c\bigg(\sum_{(i,l)\in supp(c)}a^c_{i,l}x^*_{i,l}\bigg).
	\]
	
	So there must exist $ c\in \W^* $ such that $ p^*_c>0 $ and 
	\[
	\sum_{(i,l)\in supp(c)} a^c_{i,l}x_{i,l}>\sum_{(i,l)\in supp(c)} a^c_{i,l}x^*_{i,l}.
	\] However, $ p^*_c>0 $ implies that $z^*_c=0$ ($c$ is binding at $ x^* $), and thus $ x $ violates $ c $ and is not feasible, which is a contradiction.
	
	Equal-type envy-freeness follows the fact that agents of equal type have equal consumption space and equal budgets, and face equal personalized prices.
\end{proof}

\begin{remark}\label{rmk:cheapb}
	The proof uses semi-strict quasi-concavity only in the proof of upper
	hemi-continuity of $\ul d_i$. To prove existence of an equilibrium
	without imposing the cheapest-bundle property, observe that
	continuity and quasiconcavity of $u_i$ is enough to ensure that
	$d_i$ is upper hemi-continuous, and convex- and compact-valued. If $z$ is defined from $d_i$ in place of $\ul d_i$, the proof can be written same as above. To prove that every $ \al $-slack equilibrium assignment $ x^* $ is weakly $ \C $-constrained Pareto efficient, suppose towards a contradiction that there exists a feasible assignment $ x $ such that for all $i\in I $, $ u_i(x_i)> u_i(x^*_i) $. By utility maximization, for all $ i\in I $,
	\[ p^*_i \cdot x_i> p^*_i \cdot x^*_i.
	\]
	Thus,
	\[
	\sum_{i\in I}p^*_i \cdot x_i > \sum_{i\in I}p^*_i \cdot x^*_i.
	\]So we obtain a contradiction as before.
\end{remark}

\begin{remark}[Proof of Theorem \ref{thm:existence:noend}]\label{rmk:proof:Thm1}
	The above proof can be easily adapted to prove Theorem \ref{thm:existence:noend}. We first change the price space to be $ \mathcal{P}=[0, \bar{p}]^{\Omega^*} $, where
	\begin{center}
		$ \bar{p}= \dfrac{NL}{b_{min}}+1 $, and $ b_{min}=\min\{b:(a,b)\in \Omega^*\} $.
	\end{center}
    
    By letting $ \al=1 $, Lemma \ref{lem:bcbinding} to Lemma \ref{lem:almostWL} do not change. 
    
    Lemma \ref{lem:boundary} becomes easier to prove. Suppose $ p^*_c=\bar{p} $ for some $ c\in \W^* $. Then $ z^*_c+p^*_c\ge \bar{p} $ implies that $ z^*_c\ge 0 $. So $ c $ must be binding, and for every $ (i,l)\in supp(c) $, $ p^*_{i,l}\ge a^c_{i,l} p^*_c $. However, $c$ is impossible to be binding because 
    \begin{center}
    	$ \sum_{(i,l)\in supp(c)}a^c_{i,l}x^*_{i,l} \le \sum_{(i,l)\in supp(c)}a^c_{i,l}\dfrac{1}{p^*_{i,l}}\le\sum_{(i,l)\in supp(c)}\dfrac{1}{p^*_c}\le \dfrac{NL}{p^*_c}<b_{min}$.
    \end{center}
   Then we can  prove as above that $ (x^*,p^*) $ is a pseudo-market equilibrium, and $ x^* $ is (weakly) $ \C $-constrained Pareto efficient.

\end{remark}

\section{Proof of Theorem~\ref{thm:epIR}}

Let $d_H$ denote the Hausdorff distance between two sets in $\Re^L$.
So,\[
d_H(A,B) = \max\{
\sup\{ \inf\{
\norm{x-y} : y\in B \} : x\in A \},
\sup\{ \inf\{
\norm{x-y} : x\in A \} : y\in B \} \}.\]

Let $B_i(p,\al)$ denote the budget set of agent $ i $ given a price vector $p$ and
slack $\al\in [0,1]$.  Let $\bar B_i(p,\al) = \{x_i\in\X_i:p_i\cdot x_i\le \al+(1-\al)p_i\cdot \w_i\}$ denote the budget line. Note that $B_i(p,\al)
= \{x_i\in \X_i:\exists y\in \bar B_i(p,\al) \text{ s.t. }x\leq y \}$. 

\begin{lemma}\label{lem:approx}
	For any $\da>0$, there is $\al>0$ such that if $p$ is an $ \al $-slack equilibrium price vector found in Theorem~\ref{thm:existence}, then for
	any $i$, either $p_i\cdot \w_i<1$ or 
	$d_H (\bar B_i(p,\al),\bar B_i(p,0))<\da$.
\end{lemma}

\begin{proof}
	Consider the price $\bar p$ defined in the proof of Theorem~\ref{thm:existence}. If $p$ is a price obtained in Theorem~\ref{thm:existence}, then
	$p\in [0,\bar p]^{|\W^*|}$. Note that $\bar p$ is independent of $\al$. 
	
	Let $K = \sup \{\norm{x}: x\in \X_i, 1\leq i\leq \nagents\}.$
	Now choose $\al \in
	(0,1)$ such that \[ 
	\sup\{ \abs{
		1- \frac{\al + (1-\al) p_i\cdot \w_i}{p_i\cdot \w_i}} K: p\in [0,\bar p]^{|\W^*|} 
	\text{ and } p_i\cdot \w_i\geq 1\} <\da
	\]

Observe that when $p_i\cdot \w_i\geq 1$, $B_i(p,\al)\subseteq
B_i(p,0)$. So for any $x\in B_i(p,\al)$, 
$\inf\{\norm{x-y} : y\in \bar B_i(p,0) \} = \norm{x-x} = 0$. Hence,
\[
	\sup\{ \inf\{\norm{x-y} : y\in \bar B_i(p,0) \} ,
	x\in  \bar B_i(p,\al)\}=0.\]

On the other hand, if we let $x\in \bar B_i(p,0)$, then $\g x\in \bar
B_i(p,\al)$, where 
	\[ \g =\frac{\al + (1-\al) p_i\cdot \w_i}{p_i\cdot \w_i}. \] Since
        $\g\leq 1$, $\g\in \X_i$. 
	
	Note that \[\norm{x-\g x} =\abs{1-\g} \norm{x} <\da.\] Thus 
	$\inf\{\norm{x-y} : y\in \bar B_i(p,\al) \} <\da$, and therefore
	\[
	\sup\{ \inf\{\norm{x-y} : y\in \bar B_i(p,\al) \} ,
	x\in  \bar B_i(p,0)\}<\da.\]

Thus $d_H(B_i(p,0),B_i(p,\al))<\da$.
\end{proof}

To prove the theorem, let $\da>0$ be such that, for any $p\in [0,\bar
p]^{\W^*}$, if $d_H(B_i(p,0),B_i(p,\al))<\da$ then \[ \abs{
	\max\{u_i(x):x\in B_i(p,\al)\} - \max\{u_i(x):x\in B_i(p,0) \}}<\ep.
\] For such $\da$, let $\al$ be as in Lemma~\ref{lem:approx}.

For any $i$, if $p_i\cdot \w_i<1$ then $B_i(p,0)\subseteq B_i(p,\al)$,
so \[
\max\{u_i(y):y\in \Delta_{-}\text{ and } p_i\cdot y\leq
p_i\cdot \w_i \} - u_i(x) < 0<\ep.\] If, on the contrary, $p_i\cdot \w_i\geq 1$,
then Lemma~\ref{lem:approx} implies that
$d_H(B_i(p,0),B_i(p,\al))<\da$, and the result follows from the
definition of  $\da$.

\section{Proof of Proposition~\ref{prop:envyslackwalras}}

Our first observation establishes the relation between envy and the
value of endowments at equilibrium prices. 
\begin{lemma}\label{lem:envyprices}
Let $(x,p)$ be a Walrasian equilibrium with slack $\al\in (0,1]$. If $i$ envies $j$, then $p\cdot (x_j-x_i)>0$ and $p\cdot (\w_j-\w_i)>0$.
  \end{lemma}
\begin{proof} Let $i$ envy $j$, so $u_i(x_j)>u_i(x_i)$. Then utility
  maximization implies that 
\[
\al + (1-\al) p\cdot \w_j \geq p\cdot x_j>
\al + (1-\al) p\cdot \w_i \geq p\cdot x_i,
\] where the strict inequality follows because $x_j\in\Delta_{-}$. 
So $p\cdot (x_j-x_i)>0$ and $p\cdot (\w_j-\w_i)>0$.  \end{proof}

Now consider a $\al$-slack Walrasian equilibrium $(x,p)$. 
Agent $i$'s maximization problem is:
\[
\max_{x\in \Re^L_+} u_i(x) + \la_i (I_i-p\cdot x) + \g_i (1-\one\cdot x)
\]

Where $I_i=\al+(1-\al)p\cdot\w_i$, 
$\la_i$ is a multiplier for the budget constraint, and $\g_i$ for
the $\sum_l x_{i,l}\leq 1$ constraint.

Utility functions are $C^1$. The first-order conditions for the
maximization problems are then:
 \[
\partial_l u_i(x_i) - \la_i p_l - g_i
\begin{cases}
 = 0 & \text{ if } x_{i,l}>0  \\
\leq 0 & \text{ if } x_{i,l}=0, \\
\end{cases}
\] where $\partial_l u_i(x_i) $ denotes the partial derivative of
$u_i$ with respect to $x_{i,l}$. 

Observe that if $p\cdot x_i <\al+(1-\al)p\cdot \w_i$, then the budget
constraint is not binding and $\la_i=0$. As a consequence,
$u_i(x_i)=\max\{u_i(z_i):z_i\in\Delta_{-}\}$. 
Let $S=\{i\in [N]: p\cdot x_i<\al+(1-\al)p\cdot \w_i\}$ be the set of
\df{satiated} consumers. 
Let $U=\{i\in [N]: p\cdot x_i=\al+(1-\al)p\cdot \w_i\}$ be the set of \df{unsatiated},
and observe that we can let $\la_i>0$ for all $i\in U$. 
Consider the two stage social program:

Stage 1:
\[\begin{array}{l}
\max_{\tilde y \in(\Delta_{-})^S} \sum_{i\in S} u_i(\tilde y_i)  
\end{array}
\]

Stage 2:
\[\begin{array}{l}
\max_{\tilde y \in(\Delta_{-})^U}\sum_{i\in U} \frac{1}{\la_i}u_i
(\tilde y_i)  \\
\sum_{i\in U} \tilde y_i \leq \bar\w - \sum_{i\in S} x_i
\end{array}
\]

Note that $(x_i)_{i\in S}$ solves Stage 1, while
satisfying $\sum _{i\in S}  x_i\leq \bar w$, and that given
$(x_i)_{i\in S}$, $(x_i)_{i\in U}$ solves Stage 2. That this is so
follows from the fact that $(x_i)_{i\in U}$ solves the first-order
conditions for the Stage 2 problem with Lagrange multiplier $p$ for
the constraint that $\sum_{i\in U} \tilde y_i \leq \bar\w - \sum_{i\notin S} x_i$.

Now use the assumption that $\sum _{i\in U}  x_i\gg 0$. This means
that there exists $\bar t>0$ such that if $t\in (0,\bar t]$ then the set of 
$\tilde y \in(\Delta_{-})^U$ such that $\sum_{i\in U} \tilde y_i \leq
\bar\w + t(\w_i-\w_j) -\sum_{i\notin S} x_i$ is nonempty (and, for
constraint qualification, contains an element that satisfies all
constraints with slack).  

Consider the problem 
\[\begin{array}{l}
\max_{\tilde y \in(\Delta_{-}^U)}\sum_{i\in U} \frac{1}{\la_i}u_i
(\tilde y_i)  \\
\sum_{i\in U} \tilde y_i \leq \bar\w + t(\w_i-\w_j) - \sum_{i\in S} x_i
\end{array}
\]
Note that for each $t\in (0,\bar t]$ there exists
  $(\nu(t),\g(t),\al(t))$ such that  
\[
v(t) = \sup\{
\sum_{i\in U} \frac{1}{\la_i}u_i\cdot \tilde y_i +
\nu(t) \cdot (\bar\w - \sum_{i\in S} \tilde y_i  + t (\w_i-\w_j) )-
\sum_{i\in U} \tilde y_i )
+ \sum_{i\in U} \g_i(t) (1-\sum_{l\in O} \tilde y_{i,l})
+ \sum_{i\in U} \al_i(t) \tilde y_{i,l}.
\}\] Here $\nu(t)$ is the Lagrange multiplier for the constraint that
$\sum_{i\in U} \tilde y_i \leq \bar\w - \sum_{i\in S} x_i  + t (\w_i-\w_j)$, while
$\g(t)$ and $\al(t)$ are the Lagrange multipliers for the constraint
that $(\tilde y_i)\in(\Delta_{-})^N$. Choose a selection
$(\nu(t),\g(t),\al(t))$ such that $\nu(0)=p$.  

Let $\tilde\w=\bar\w - \sum_{i\in S} x_i$. The saddle point
inequalities imply that 
\begin{align*} (t'-t) \nu(t)\cdot (\w_i-\w_j) & = 
\sum_{i\in U} \frac{1}{\la_i}u_i(x_i(t') )
 +\nu(t) \cdot (\tilde\w + t' (\w_i-\w_j) - \sum_{i\in U} x_i(t') ) \\
& + \sum_{i\in U} \g_i(t) (1-\sum_{l\in O} x_{i,l}(t')) 
 + \sum_{i\in U} \al_i(t) x_{i,l}(t') \\
& -\left(
\sum_{i\in U} \frac{1}{\la_i}u_i(x_i(t'))
 +\nu(t) \cdot (\tilde\w + t (\w_i-\w_j) - \sum_{i\in U} x_i(t') ) \right.  \\
& + \sum_{i\in U} \g_i(t) (1-\sum_{l\in O} x_{i,l}(t')) 
 \left. + \sum_{i\in U} \al_i(t) x_{i,l}(t') \right)\\
& \geq v(t') - v(t)
\end{align*} 

Now recall that $\nu(0)=p$. Then Lemma~\ref{lem:envyprices}, together
with the above inequality, imply that 
\[
0> p\cdot (\w_i-\w_j) t' \geq v(t') - v(0) \] for all $t'>0$ with
$t'\leq \bar t$. 

\footnotesize

\bibliographystyle{econometrica}
\bibliography{envy}

\end{document}